\documentclass[acmsmall,nonacm]{acmart}
\acmJournal{PACMPL}
\acmArticle{759}
\acmYear{2025}
\acmMonth{7}
\acmDOI{} 
\startPage{1}

\setcopyright{none}

\usepackage{booktabs}   
\usepackage{subcaption} 
\usepackage[percent]{overpic}
\usepackage{xspace}
\usepackage{amsthm}
\usepackage[linesnumbered]{algorithm2e}
\usepackage{empheq}
\usepackage{amsmath}
\usepackage{mathtools}
\usepackage{pifont}
\usepackage{array}
\usepackage{multirow}
\usepackage{wrapfig}

\usepackage{listings}

\newtheorem*{theorem*}{Theorem}
\newtheorem{theorem}{Theorem}

\newtheorem*{lemma*}{Lemma}
\newtheorem{lemma}{Lemma}

\definecolor{lbcolor}{rgb}{0.875,0.875,0.875}
\newcommand{\rlibm}{\textsc{RLibm}\xspace}
\newcommand{\tool}{\textsc{RLibm-MultiRound}\xspace}

\newcommand{\eg}{\emph{e.g.}\xspace}
\newcommand{\ie}{\emph{i.e.}\xspace}
\newcommand{\etal}{\emph{et al.}\xspace}

\newcommand{\RNE}{\textit{RN}\xspace}
\newcommand{\RNZ}{\textit{RZ}\xspace}
\newcommand{\RNP}{\textit{RU}\xspace}
\newcommand{\RNN}{\textit{RD}\xspace}

\newcommand{\RZA}{\textit{RZA}\xspace}
\newcommand{\RZM}{\textit{RZM}\xspace}

\newcommand{\RNO}{\textit{RO}\xspace}

\newcommand{\cmark}{\ding{51}}
\newcommand{\xmark}{\textcolor{red}{\ding{55}}}
\newcommand{\NA}{N/A}

\newcommand{\PreserveBackslash}[1]{\let\temp=\\#1\let\\=\temp}
\newcolumntype{C}[1]{>{\PreserveBackslash\centering}p{#1}}

\begin{document}

\title[Correctly Rounded Math Libraries Without Worrying about
  the Application's Rounding Mode]{\tool: Correctly Rounded Math
  Libraries Without Worrying about the Application's Rounding Mode
  \\ {\small \textbf{Rutgers Department of Computer Science Technical
      Report DCS-TR-759}}}

\author{Sehyeok Park}
\orcid{0009-0002-1528-562X}
\affiliation{%
  \institution{Rutgers University}
  \city{Piscataway}
  \country{USA}
}
\email{sp2044@cs.rutgers.edu}

\author{Justin Kim}
\orcid{0009-0002-1481-5019}
\affiliation{%
  \institution{Rutgers University}
  \city{Piscataway}
  \country{USA}
}
\email{jk1849@scarletmail.rutgers.edu}

\author{Santosh Nagarakatte}
\orcid{0000-0002-5048-8548}
\affiliation{%
  \institution{Rutgers University}
  \city{Piscataway}
  \country{USA}
}
\email{santosh.nagarakatte@cs.rutgers.edu}

\begin{abstract}
Our \rlibm project has recently proposed methods to generate a single
implementation for an elementary function that produces correctly
rounded results for multiple rounding modes and representations with
up to 32-bits. They are appealing for developing fast reference
libraries without double rounding issues. The key insight is to build
polynomial approximations that produce the correctly rounded result
for a representation with two additional bits when compared to the
largest target representation and with the ``non-standard''
round-to-odd rounding mode, which makes double rounding the \rlibm
math library result to any smaller target representation innocuous.
The resulting approximations generated by the \rlibm approach are
implemented with machine supported floating-point operations with the
\emph{round-to-nearest} rounding mode. When an application uses a
rounding mode other than the round-to-nearest mode, the \rlibm math
library saves the application’s rounding mode, changes the system’s
rounding mode to round-to-nearest, computes the correctly rounded
result, and restores the application's rounding mode. This frequent
change of rounding modes has a performance cost.

This paper proposes two new methods, which we call rounding-invariant
outputs and rounding-invariant input bounds, to avoid the frequent
changes to the rounding mode and the dependence on the
round-to-nearest mode. First, our new rounding-invariant outputs
method proposes using the round-to-zero rounding mode to implement
\rlibm's polynomial approximations. We propose fast, error-free
transformations to emulate a round-to-zero result from any standard
rounding mode without changing the rounding mode. Second, our
rounding-invariant input bounds method factors any rounding error due
to different rounding modes using interval bounds in the \rlibm
pipeline. Both methods make a different set of trade-offs and improve
the performance of resulting libraries by more than 2$\times$.
\end{abstract}

\maketitle
\section{Introduction}

Math libraries provide implementations of commonly used elementary
functions. The outputs of these elementary functions are irrational
values for almost all inputs and cannot be represented exactly in a
finite precision floating-point~(FP) representation.  The correctly
rounded result of an elementary function for a given input is the
result produced after computing the result with infinite precision and
then rounded to the target representation. The problem of generating
correctly rounded results for arbitrary target representations is
known to be challenging (\ie, Table Maker's
dilemma~\cite{Kahan:tablemaker:online:2004}). Hence, the IEEE-754
standard recommends but does not mandate correctly rounded results for
elementary functions.
Correctly rounded math libraries enable portability and
reproducibility of applications using them.

Recent efforts, such as the CORE-MATH
project~\cite{sibidanov:core-math:arith:2022} and our \rlibm
project~\cite{lim:rlibmall:popl:2022,lim:rlibm:popl:2021,lim:rlibm32:pldi:2021},
have demonstrated that fast and correctly rounded libraries are
feasible. There is also a working group discussion to require
correctly rounded implementations in the upcoming 2029 IEEE-754
standard~\cite{brisebarre:crf-what-cost:2024}.
The minimax approximation method is the most well-known method for
building correctly rounded libraries. Effectively, these methods
generate polynomial approximations that minimize the maximum error
across all inputs with respect to the real value~(see Chapter~3 of
\cite{Muller:elemfunc:book:2016}). Subsequently, the error in the
polynomial evaluation methods are bounded to ensure that numerical
errors do not change the rounding decision.

\begin{figure}
  \centering{\includegraphics[width=0.90\textwidth]{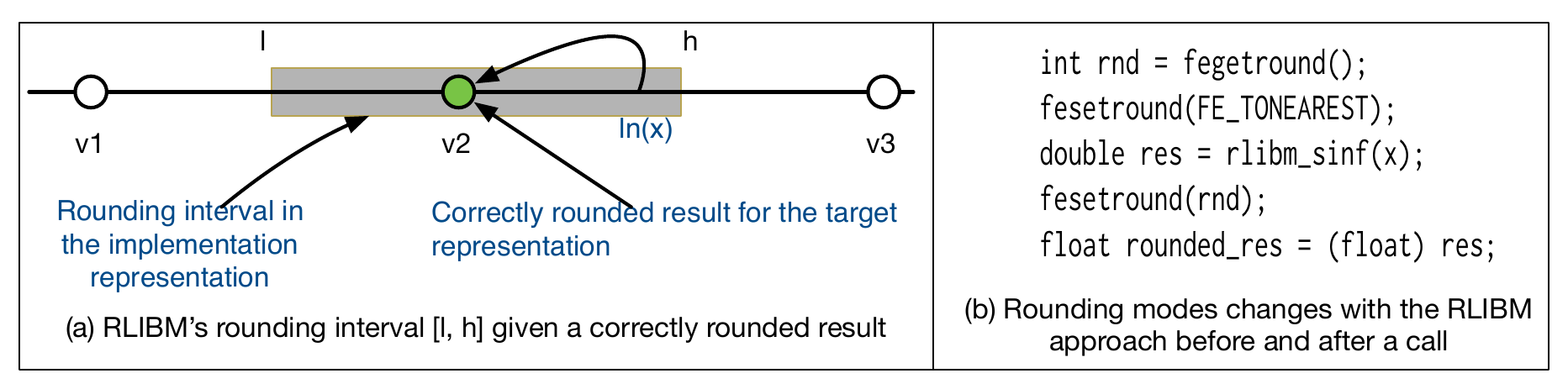}}
  \caption{\small (a) The three representable FP values (v1, v2, and
    v3) in the target representation. \rlibm's rounding interval $[l,
      h]$ (shown in gray) in the implementation representation such
    that any value in this interval rounds to v2 assuming the rounding
    mode is round-to-nearest-ties-to-even. (b) Changes to the rounding
    mode performed by the \rlibm approach before and after calling the
    elementary function.}
  \label{fig:intro}
\end{figure}

\textbf{Our \rlibm project.} Unlike traditional minimax methods, our
\rlibm project makes a case for directly approximating the correctly
rounded
result~\cite{lim:rlibmall:popl:2022,lim:rlibm:popl:2021,lim:rlibm32:pldi:2021}. The
insight is to split the task of generating the oracle and the task of
generating an efficient implementation given an oracle such as the
MPFR library~\cite{Fousse:toms:2007:mpfr}.
When building a correctly rounded library for the 32-bit FP
representation (\ie, the target representation), the \rlibm project
implements the library using the 64-bit FP representation (\ie, the
implementation representation). Then, there is an interval of values
in the implementation representation around the correctly rounded
result of the target representation such that any value in this
interval rounds to the correctly rounded result.
Figure~\ref{fig:intro}(a) shows the correctly rounded result (\ie, v2)
and the rounding interval. Given this interval $[l,h]$, the task of
producing a correctly rounded result for an input $x$ with a
polynomial of degree $d$ can be expressed as a linear constraint: $l
\leq C_0 + C_1 x + C_2 x^2 + C_3 x^3 + \ldots + C_d x^d \leq h$.  The
size of the rounding interval is 1 ULP (units in the last place) for
all inputs.
The freedom available with the \rlibm approach is significantly larger
than the freedom available with minimax methods.
Hence, the \rlibm project generates fast low-degree
polynomials.

\textbf{Multiple representations and rounding modes.}  Low precision
representations are becoming mainstream especially with
accelerators~(\eg, \texttt{bfloat16}~\cite{Wang:tpu:online:2019},
\texttt{tensorfloat32}~\cite{nvidia:tensorfloat:online:2020}, and
\texttt{FP8}). They are increasingly used in scientific computing
apart from machine learning. Further, the IEEE-754 standard specifies
four distinct rounding modes for the binary FP representation:
round-to-nearest-ties-to-even~(\RNE), round-towards-zero~(\RNZ),
round-up~(\RNP), and round-down~(\RNN). Each rounding mode is
attractive in specific domains.
For example, the computational geometry algorithms library (CGAL) uses
different rounding modes.
Similarly, some hardware accelerators use the round-to-zero~(\RNZ)
mode because it can be implemented efficiently.
Rather than designing a custom math library for each such rounding
mode and representation, generating a single math library that handles
all these rounding modes and new representations is attractive as a
reference library.
Existing correctly rounded libraries for a single
representation such as CORE-MATH~\cite{sibidanov:core-math:arith:2022}
and CR-LIBM~\cite{Daramy:crlibm:spie:2003} do not produce correctly
rounded results when they are repurposed for these new representations
because of double rounding
errors~\cite{david:pitfalls:toplas:2008}. The first rounding happens
when the real value is rounded to the representation the math library
was originally designed for and the second rounding happens when the
result from the math library is rounded to the target representation.

\textbf{\rlibm's method to handle multiple representations.} The
\rlibm project includes an appealing method to generate a single
implementation that produces correctly rounded results for multiple
representations and rounding modes~\cite{lim:rlibmall:popl:2022}. When
the goal is to generate correctly rounded results for all
representations with up to $n$-bits, the \rlibm project's approach is to
approximate the correctly rounded result of a $(n+2)$-bit
representation with a non-standard rounding mode called
\texttt{round-to-odd}. In the round-to-odd mode, the real value that
is not exactly representable is rounded to the nearest FP value whose
bit-pattern is odd. When a real value is exactly representable in the
FP representation, it is represented with that FP value. When the
round-to-odd result with the $(n+2)$-bit representation is subsequently
rounded to any target representation with $n$ or fewer bits, it
produces correctly rounded results. Effectively, \rlibm's approach of
computing the round-to-odd (\RNO) result with a $(n+2)$-bit
representation makes double rounding harmless.

\textbf{Range reduction, output compensation, and polynomial
  evaluation with FP arithmetic.} Typically, polynomial approximations
are feasible over domains much smaller than the dynamic range of a
32-bit FP representation. For example, it is much more effective to
approximate $\mathit{log}(x)$ over inputs $x \in [0, 1/128)$ rather
  than over the entire range of 32-bit floats where $|x| \in
  [2^{-149}, 2^{128})$.
As a first step, each input $x$ in the original domain is transformed
to a value in a smaller domain $x'$ through a process known as range
reduction.  The \rlibm implementations represent each range reduced
input $x'$ as a 64-bit, double-precision FP number, which is the
internal representation used for all subsequent FP
operations. Subsequently, a polynomial approximation computes the
result for the input in the small domain (\ie, $y' =
P(x')$). Additional operations that are collectively known as the
output compensation function map the result $y'$ to produce the result
for the original input~(\ie, $y = OC(y', x)$). Range reduction,
polynomial evaluation, and output compensation are performed with FP
operations and can accumulate rounding error.

To generate a polynomial approximation with guaranteed correctness,
the \rlibm pipeline first computes the \emph{reduced input} $x'$ for
each input $x$ with the range reduction algorithm. Using the 34-bit
\RNO oracle of $f(x)$, it computes the round-to-odd rounding interval
$[l, h]$ for every input. The \rlibm pipeline subsequently identifies
for each reduced input the widest possible \emph{reduced interval}
$[l', h']$ such that $\forall{y'} \in [l', h'], l \leq OC(y', x) \leq
h$. Finally, it solves a system of linear inequalities $l' \leq P(x')
\leq h'$ to generate the polynomial approximation $P(x')$. Given the
manner in which each reduced interval $[l', h']$ is derived (\ie,
$\forall{y'} \in [l', h'], l \leq OC(y', x) \leq h$), a polynomial
evaluation result that satisfies $l' \leq P(x') \leq h'$ is guaranteed
to satisfy $l \leq OC(P(x'), x) \leq h$.

\textbf{\rlibm uses round-to-nearest as the implementation rounding
  mode.} The \rlibm project uses the \RNE mode for its
implementations. When an application uses a rounding mode other than
\RNE, the \rlibm project saves the application's rounding mode,
changes the default rounding mode of the system to \RNE, computes the
output of the math library implementation, and restores the
application's rounding mode before the final rounding to the target
representation as shown in Figure~\ref{fig:intro}(b).
 Each rounding mode change can incur up to 40 cycles on a modern Linux
 machine.
Specifically, \rlibm implementations are only guaranteed to produce
intermediate values within the rounding intervals of the final results
when they are invoked using the \RNE mode.
Hence, not changing the rounding mode to \RNE as shown in
Figure~\ref{fig:intro}(b) may lead to the wrong results.

\textbf{This paper.} This paper proposes two new methods, which we
call rounding-invariant outputs and rounding-invariant input bounds,
to completely eliminate the rounding mode changes necessitated by the
\rlibm approach while maintaining correctness under all
application-level rounding modes.

\textbf{Rounding-invariant outputs by emulating round-to-zero
  results.} In our rounding-invariant outputs method, we propose to
use round-to-zero~(\RNZ) as the underlying implementation rounding
mode instead of round-to-nearest~(\RNE). Our key insight is that it is
possible to emulate the \RNZ result under any of the four rounding
modes without having to explicitly change the application's rounding
mode. We design new algorithms that compute the \RNZ result
irrespective of the application's rounding mode using error-free
transformations that adjust the error in an FP operation using a
sequence of auxiliary FP operations and bit manipulations (see
Section~\ref{subsec:round-to-zero}).
This method requires very few changes to the \rlibm pipeline except
performing reduced interval and polynomial generation with the \RNZ
mode. However, it requires wrapping every rounding mode-dependent
addition and multiplication in the final implementations with our new
algorithms to adjust the initial results, which is faster than
changing the rounding mode but still entails noticeable overhead.

\textbf{Rounding-invariant input bounds for measuring variability
  induced by various rounding modes.}  The key idea is to bound the
range of values that can arise across different rounding modes given
the different possible combinations of rounding error associated with
a series of FP operations~(see Section~\ref{sec:multi-round}). Under
this new approach, we no longer consider the final output of a
sequence of FP operations to be a single FP value, as it would be if
all operations adhered to a single rounding mode. Instead, we treat
the result as an interval to account for the varying effects of
different rounding rules.
This approach is based on the property that for any given faithfully
rounded FP arithmetic operation on finite operands (\ie, operands that
are neither infinity nor $NaN$), the round-down ($\RNN$) result is the
lower bound and the round-up ($\RNP$) result is the upper bound.
Leveraging this property,
we compose the bounds on the result of each FP operation bottom-up
using interval arithmetic to deduce the final bounds for a target
sequence. In doing so, we can identify the minimum and maximum
possible outputs a candidate implementation could ultimately produce
for a given input across all rounding modes and confirm that both
values satisfy the associated correctness constraints. This approach
requires non-trivial changes to the \rlibm pipeline as it involves
correctness constraints that are different from those derivable from a
single rounding mode. The main advantage of this approach is that the
outputs of the FP operations in the resulting implementations do not
require any adjustments to satisfy correctness.

Our resulting library is more than $2\times$ faster than the existing
\rlibm prototypes. Our prototype is the first math library that
produces correctly rounded results for all inputs across multiple
representations for all four standard rounding modes, regardless of
the application-level rounding mode.

\section{Background}
\label{sec:background}
\begin{wrapfigure}{r}{9cm}
  \centering{\includegraphics[width=0.6\textwidth]{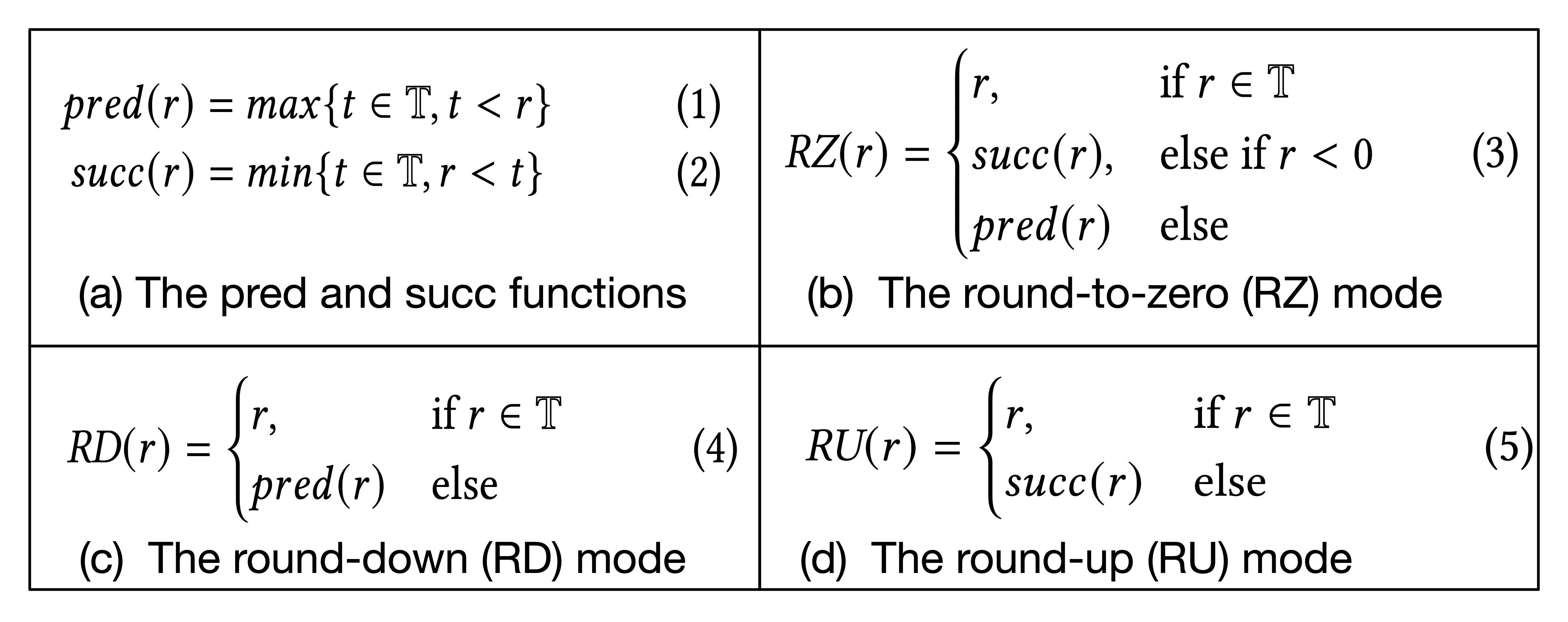}}
  \caption{\small (a) The pred and succ functions used for faithful
    rounding of a real number $r$. The \RNZ, \RNN, and \RNP rounding
    modes defined using the pred and succ functions are shown in (b),
    (c), and (d), respectively.}
   \label{fig:background:modes}
\end{wrapfigure}

\textbf{Rounding modes.}
The IEEE-754 standard provides the specification for four rounding
modes for the binary FP representation. These are
round-to-nearest~($\mathit{\RNE}$), round-to-zero~($\mathit{\RNZ}$),
round-up~($\mathit{\RNP}$), and round-down~($\mathit{\RNN}$). Since
this paper is about implementing math libraries that can operate
directly under any application level rounding mode, we provide
background to understand the behavior of various rounding modes and
their interactions with FP arithmetic.

Let $\mathbb{R}$ represent the set of all real numbers and
$\mathbb{T}$ represent the set of all the numbers in a given FP
representation.  With respect to the outputs and operands of FP
operations in the \rlibm project, $\mathbb{T}$ is the set of all
64-bit FP numbers. For a given number $r \in \mathbb{R}$, its
neighbors in the target representation $\mathbb{T}$, denoted $pred(r)$
and $succ(r)$, can be defined through the equations in
Figure~\ref{fig:background:modes}(a).

For a given rounding function to be faithful, it must return either
$\mathit{pred}(r)$ and $\mathit{succ}(r)$ whenever the input $r$ is
not exactly representable in $\mathbb{T}$. The four rounding modes
considered by \rlibm, which are $\mathit{\RNZ}$, $\mathit{\RNE}$,
$\mathit{\RNN}$, and $\mathit{\RNP}$, all adhere to this
requirement. We apply the notation $\mathit{rnd}(r)$ to denote a
rounding function that applies any of these four rounding modes. For
the purposes of this paper, we restrict the domain of all functions
$\mathit{rnd}(r)$ to \emph{non-zero real numbers}.

\textbf{The round-to-zero~(\RNZ) mode}. For our rounding-invariant
outputs approach, we propose to use $\mathit{\RNZ}$ as the default
implementation rounding mode and simulate its result across all
application rounding modes. The rounding function $\mathit{\RNZ}(r)$,
which applies rounding via $\mathit{\RNZ}$, is defined as shown in
Figure~\ref{fig:background:modes}(b).

\textbf{The round-down~(\RNN) and the round-up~(\RNP) mode.}  Our
proposed rounding-invariant input bounds method uses the
$\mathit{\RNN}$ and $\mathit{\RNP}$ modes to bound the variability
induced by the various rounding modes. The rounding functions
$\mathit{\RNN}(r)$ and $\mathit{\RNP}(r)$ can be defined as shown in
Figure~\ref{fig:background:modes}(c) and
Figure~\ref{fig:background:modes}(d), respectively.

The definitions of $\mathit{pred}(r)$ and $\mathit{succ}(r)$ in
Figure~\ref{fig:background:modes}(a) (Equations~1 and 2) along with
Equations~4 and 5 in Figure~\ref{fig:background:modes} reveal the
following properties for $\mathit{\RNN}$ and $\mathit{\RNP}$:
$\forall{r} \in \mathbb{R} \setminus \{0\}, \mathit{pred}(r) \leq
\mathit{\RNN}(r)$ and $\forall{r} \in \mathbb{R} \setminus \{ 0 \},
\mathit{\RNP}(r) \leq \mathit{succ}(r)$. Given these properties, one
could define $\mathit{\RNN}$ and $\mathit{\RNP}$ in the following
manner.

\begin{definition}\label{rd-def-alt}
For all $r \in \mathbb{R} \setminus \{0\}$, $\mathit{\RNN}(r)$ is the largest number $t \in \mathbb{T}$ such that $t \leq r$.
\end{definition}

\begin{definition}\label{ru-def-alt}
For all $r \in \mathbb{R} \setminus \{0\}$, $\mathit{\RNP}(r)$ is the smallest number  $t \in \mathbb{T}$ such that $t \geq r$.
\end{definition}

Using the properties and definitions of faithful rounding,
$\mathit{\RNN}$, and $\mathit{\RNP}$, we state the following lemma
providing the expected bounds on the faithfully rounded outputs of
rounding functions.

\begin{lemma}\label{lemma:rnd-bound}
Let $\mathit{rnd}$ be any rounding function that faithfully rounds a
number $r \in \mathbb{R} \setminus \{ 0 \}$ to a number $t \in
\mathbb{T}$. $\forall r \in \mathbb{R} \setminus \{0\}, \mathit{\RNN}(r)\leq
\mathit{rnd}(r) \leq \mathit{\RNP}(r)$.
\end{lemma}

The lemma is directly derivable from the definition of faithful
rounding. Lemma~\ref{lemma:rnd-bound} guarantees that when a non-zero
real number $r$ is rounded using a faithful rounding function
$\mathit{rnd}$, $\RNN(r)$ and $\mathit{\RNP}(r)$ will respectively serve as the
lower and upper bounds of $\mathit{rnd}(r)$.
Lemma~\ref{lemma:rd-monotonic} and ~\ref{lemma:ru-monotonic} detail
the well-established monotonically non-decreasing properties of
faithful rounding functions~\cite{muller:fp:2018}, specifically with
regard to $\RNN$ and $\RNP$.

\begin{lemma}\label{lemma:rd-monotonic}
$\forall{a}, \forall{b} \in \mathbb{R} \setminus \{0\}, a \leq b \implies \RNN(a) \leq \RNN(b)$  
\end{lemma}

\begin{lemma}\label{lemma:ru-monotonic}
$\forall{a}, \forall{b} \in \mathbb{R} \setminus \{0\}, a \leq b \implies \RNP(a) \leq \RNP(b)$
\end{lemma}

\textbf{Preservation of signs with faithful rounding.}
The final property of interest pertains to the preservation of signs.

\begin{definition}\label{rnd-sign}
 For all $v \in (\mathbb{R} \setminus \{0\}) \cup \mathbb{T}$ where
 $\mathbb{T}$ is a FP representation, we define $\mathit{sign}(v)$ to
 be $0$ for positive numbers and $1$ for negative numbers. For the FP
 numbers $+0,-0 \in \mathbb{T}$, we define the sign as
 $\mathit{sign}(+0) = 0$ and $\mathit{sign}(-0) = 1$.
\end{definition}

\begin{lemma}\label{lemma:sign-preservation}
Let $\mathit{rnd}$ be any rounding function that faithfully rounds a
number $r \in \mathbb{R} \setminus \{ 0 \}$ to a FP number $t \in
\mathbb{T}$. For all $r \in \mathbb{R}\setminus\{0\}$ and for all
$\mathit{rnd}$, $\mathit{sign}(r) = \mathit{sign}(\mathit{rnd}(r))$.
\end{lemma}

Under our definition of $\mathit{sign}$,
Lemma~\ref{lemma:sign-preservation} signifies the sign preserving
property of faithful rounding for non-zero values.  We provide the
proof for Lemma~\ref{lemma:sign-preservation} in the supplemental
material~(Section~7).

Having introduced pertinent properties of faithful rounding and our
definition of $\mathit{sign}(r)$, we refer back to Equations~3 through
5 and elaborate upon crucial intricacies. We constrain the domains of
the rounding functions $\mathit{rnd}(r)$ for all $\mathit{rnd} \in \{
\mathit{\RNE}, \mathit{\RNZ}, \mathit{\RNN}, \mathit{\RNP} \}$ to
non-zero real numbers. This is because all FP representations
$\mathbb{T}$ considered in this paper treat $+0$ and $-0$ as separate
FP numbers while equating both to $0$ in the context of real
arithmetic. Distinguishing $+0$ and $-0$ creates ambiguity as to which
of the two numbers a rounding function should return for 0. This
ambiguity can be resolved for rounded FP arithmetic by defining what
an FP operation $a \odot b$ should return when its real number
counterpart $ a \cdot b = 0$. For the FP operations of concern, which
are addition and multiplication, the choice between $+0$ or $-0$ is
dependent on both the rounding mode environment and the sign of the
operands. We provide IEEE-754 standard-compliant definitions of $a
\oplus_{\mathit{rnd}} b$ and $a \otimes_{\mathit{rnd}} b$ in
Figure~\ref{fig:background:ops}, which represent the output of a FP
addition and multiplication under a given rounding mode $\mathit{rnd}
\in \{ \mathit{\RNE}, \mathit{\RNZ}, \mathit{\RNN},
\mathit{\RNP}\}$. Henceforth, we reserve the notations $\oplus$ and
$\otimes$ for FP addition and multiplication respectively while using
$+$ and $\times$ solely for real arithmetic operations. When the
rounding rule being applied is relevant, we apply the notations
$\oplus_{\mathit{rnd}}$ and $\otimes_{\mathit{rnd}}$.  We restrict the
domain of the equations for $\oplus_{\mathit{rnd}}$ and
$\otimes_{\mathit{rnd}}$ to non-$NaN$, non-infinity operands.

\begin{figure}
  \centering{\includegraphics[width=0.95\textwidth]{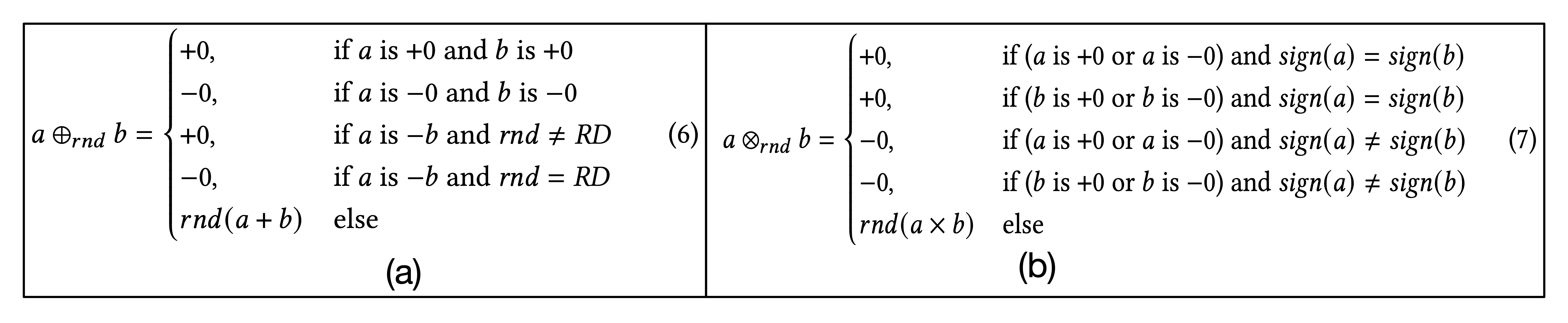}}
  \caption{\small (a) The rounded addition
      $\oplus_{rnd}$ for any rounding mode $\mathit{rnd} \in \{
      \mathit{RN}, \mathit{RZ}, \mathit{RD}, \mathit{RU} \}$. (b) The
      rounded multiplication $\otimes_{\mathit{rnd}}$ for any rounding
      mode $\mathit{rnd}$. }
   \label{fig:background:ops}
\end{figure}

In the context of Equation~6 in Figure~\ref{fig:background:ops}(a),
$\mathit{rnd}(a+b)$ represents the output obtained from subjecting the
real arithmetic result of $a+b$ to a rounding function
$\mathit{rnd}(r)$ as exemplified in Equations 3 through 5. The
expression $\mathit{rnd}(a \times b)$ in Equation~7 can be interpreted
analogously to $\mathit{rnd}(a + b)$. The first four cases of
Equation~6 detail different scenarios under which $a+b = 0$. Likewise,
the first four cases of Equation~7 cover situations where $a \times b
= 0$. Based on these equations, we infer that $a \oplus_{\mathit{rnd}}
b = rnd(a+b)$ whenever $a+b \neq 0$ and $a \otimes_{\mathit{rnd}} b = \mathit{rnd}(a
\times b)$ whenever $a \times b \neq 0$.
By construction, $a \oplus_{\mathit{rnd}} b$ and $a
\otimes_{\mathit{rnd}} b$ return faithfully rounded versions of $a+b$
and $a \times b$ respectively. These operations therefore possess the
properties of faithful rounding detailed in Lemma 1 as it pertains to
$\mathit{\RNN}$'s and $\mathit{\RNP}$'s roles in producing the lower
and upper bounds respectively. The two operations also reflect the
monotonic properties of rounding detailed in
Lemmas~\ref{lemma:rd-monotonic} and ~\ref{lemma:ru-monotonic}. We
highlight the properties of $\oplus_{\mathit{rnd}}$ and
$\otimes_{\mathit{rnd}}$ most important to our theorems through the
following lemmas.

\begin{lemma}\label{lemma:rnd-sum-bound}
$\forall{a,b} \in \mathbb{T} \setminus \{ NaN, \pm \infty \}, \forall{\mathit{rnd}} \in \{ \RNE, \RNZ, \RNN, \RNP \}, a \oplus_{\RNN} b \leq a \oplus_{\mathit{rnd}} b \leq a \oplus_{\RNP} b$.  
\end{lemma}

\begin{lemma}\label{lemma:rnd-prod-bound}
$\forall{a,b} \in \mathbb{T} \setminus \{ NaN, \pm \infty \}, \forall{\mathit{rnd}} \in \{ \RNE, \RNZ, \RNN, \RNP \}, a \otimes_{\RNN} b \leq a \otimes_{\mathit{rnd}} b \leq a \otimes_{\RNP} b$.  
\end{lemma}

\begin{lemma}\label{lemma:rnd-sum-monotonic}
$\forall{a,b,c,d} \in \mathbb{T} \setminus \{ NaN, \pm \infty \}, a+b \leq c+d \implies (a \oplus_{\RNN} b \leq  c \oplus_{\RNN} d) \land (a \oplus_{\RNP} b \leq  c \oplus_{\RNP} d)$.  
\end{lemma}

\begin{lemma}\label{lemma:rnd-prod-monotonic}
$\forall{a,b,c,d} \in \mathbb{T} \setminus \{ NaN, \pm \infty \}, a \times b \leq c \times d \implies (a \otimes_{\RNN} b \leq  c \otimes_{\RNN} d) \land (a \otimes_{\RNP} b \leq  c \otimes_{\RNP} d)$.  
\end{lemma}

\textbf{Propagating bounds using interval arithmetic.}
Our rounding-invariant input bounds approach treats both the operands
and output of each FP operation as a \emph{range} of values rather
than a single value. It relies on the properties of interval
arithmetic listed below to identify the expected lower and upper
bounds on the result of a FP operation when its operands are
represented as ranges of FP numbers excluding $NaN$s and $\pm \infty$.
Specifically, we use these properties to identify the output bounds of
the \emph{real arithmetic} counterpart of an FP operation given the
ranges of its FP operands. We denote the lower and upper bounds of the
final output $a$ of an ordered sequence of FP operations as
$\underline{a}$ and $\overline{a}$, respectively.

\begin{lemma}\label{lemma:int-sum}
$\forall{a\in[\underline{a},\overline{a}]},\forall{b\in[\underline{b},\overline{b}]}, \underline{a}+\underline{b} \leq a+b \leq \overline{a}+\overline{b}$.
\end{lemma}

\begin{lemma}\label{lemma:int-prod}
\begin{tabbing}
\hspace{1.9cm} $\forall{a\in[\underline{a},\overline{a}]},\forall{b\in[\underline{b},\overline{b}]}$,
\= $\min{(\underline{a} \times \underline{b}, \underline{a} \times \overline{b}, \overline{a} \times \underline{b}, \overline{a} \times \overline{b})} \leq a \times b$ $\wedge$ \\
\> $a \times b \leq   \max{(\underline{a} \times \underline{b}, \underline{a} \times \overline{b}, \overline{a} \times \underline{b}, \overline{a} \times \overline{b})}$. \
\end{tabbing}
\end{lemma}

\section{Rounding Mode Independence Using Our Approach}

Our goal is to develop implementations for elementary functions that
produce correctly rounded results for all FP inputs across multiple
representations with up to $32$-bits.
We seek to achieve correctness for all four standard rounding
modes~(\eg, \RNE, \RNZ, \RNP, and \RNN) and any faithful rounding mode
potentially used by the invoking applications and to do so without
requiring any explicit rounding mode changes.
Using the \rlibm approach, we attempt to generate polynomial
approximations over the reduced inputs, which when used with the
output compensation function can produce 64-bit values within the
rounding interval of the 34-bit round-to-odd (\RNO) result.
When a 64-bit FP value in the rounding interval of the 34-bit \RNO
result is \emph{double rounded} to any target representation less than
or equal to 32-bits, it is guaranteed to produce the correctly rounded
result regardless of the rounding rule used for the final
rounding.

The main task is thus ensuring that a candidate approximation can
produce such 64-bit values for all inputs regardless of the invoking
application's rounding mode.
The key challenge in this endeavor is that the range reduction,
polynomial evaluation, and output compensation processes involve FP
arithmetic, which can experience different rounding errors depending
on the rounding mode. Hence, an implementation's 64-bit intermediate
output prior to the final rounding may differ depending on the
rounding mode under which it was produced.
The \rlibm pipeline generates candidate approximations that produce
correct results with round-to-nearest, which is the default rounding
mode used by its generators. Inevitably, the \rlibm implementations
necessitate rounding mode changes to \RNE to ensure correctness.
Saving and restoring the application's rounding mode as required by
the \rlibm prototype requires up to 40 cycles for each input.

This paper proposes two new methods to completely remove the rounding
mode changes required by the \rlibm approach. First, we make a case
for using the round-to-zero (\RNZ) mode as the implementation rounding
mode. We design error-free transformations that compute the error in a
FP operation to simulate the \RNZ mode result when the application
uses any other rounding mode (see
Section~\ref{subsec:round-to-zero}). Our second method computes the
bounds on the full range of outputs that are possible when evaluating
polynomial approximations and output compensation functions under
different rounding modes. We propose a new method that recursively
defines the lower and upper bounds for individual FP operations to
account for rounding errors across different rounding modes and
composes them bottom-up using interval arithmetic ~(see
Section~\ref{sec:multi-round}) to represent the results obtained from
evaluating polynomials and output compensation
functions. Subsequently, we change the reduced interval and polynomial
generation processes in the \rlibm pipeline to generate polynomials
that satisfy the new constraints for correctness stemming from the
incorporation of rounding induced variability.

Both these approaches have their own trade-offs. The first method
requires augmenting every FP operation with steps that collectively
compute its rounding error, assess whether its output conforms to
$\RNZ$, and make necessary adjustments. In contrast, the second method
requires extensive changes to the \rlibm pipeline with regard to
generating reduced intervals and polynomials.
It does not require any additional operations for evaluating
polynomials in the resultant implementations once the polynomial
approximations are generated.

\subsection{Rounding Independence with Round-to-Zero Emulation}
\label{subsec:round-to-zero}

To produce the $\RNZ$ result across all rounding modes of interest, we
design error-free transformations, which are a sequence of FP
operations that compute the error in a given operation using FP
arithmetic. Using the error-free transformations, we design a decision
procedure that determines whether the original result produced under
the application's rounding mode needs to be adjusted to match the
result expected under $\RNZ$.

\textbf{Intuition for computing the \RNZ result.} From the perspective
of implementing \rlibm's core elementary functions with 64-bit FP
operations, the operations of concern are addition and
multiplication. Given two non-$NaN$, non-infinity 64-bit FP numbers
$a$ and $b$, we want to compute $a \oplus_{\mathit{\RNZ}} b$ and $a
\otimes_{\mathit{\RNZ}} b$ whenever $a \oplus_{\mathit{rnd}} b$ and $a
\otimes_{\mathit{rnd}} b$ are computed with different rounding
modes. The key challenge in computing the $\mathit{\RNZ}$ results is
that neither $a \oplus_{\mathit{rnd}} b$ nor $a \otimes_{\mathit{rnd}}
b$ provides any direct indication as to whether $a
\oplus_{\mathit{rnd}} b \neq a + b$ or $a \otimes_{\mathit{rnd}} b
\neq a \times b$. The goal is thus to create a sequence of FP
operations that compute the rounding error in $a \oplus_{\mathit{rnd}}
b$ (or $a \otimes_{\mathit{rnd}} b$) with respect to the real result
$a + b$ (or $a \times b$) and adjust the original output to match the
$\mathit{\RNZ}$ result based on the error. We preface the discussion
of our algorithms by emphasizing that they are intended to target
non-$\mathit{NaN}$, non-infinity FP operands for which neither the sum
nor the product induces overflow. The FP numbers used with these
algorithms in the final implementations satisfy these criteria.

\textbf{Relationship between $\oplus_{\mathit{RZ}}$ and other
  $\oplus_{\mathit{rnd}}$.}  Given two FP numbers $a$ and $b$ that
need to be summed, we need to adjust the output $a
\oplus_{\mathit{rnd}} b \in \mathbb{T}$ whenever $a
\oplus_{\mathit{rnd}} b$ differs from $a \oplus_{\mathit{\RNZ}}
b$. Given Equations~3 (Figure~\ref{fig:background:modes}) and~6
(Figure~\ref{fig:background:ops}), this occurs when either $a$ is $-b$
and $\mathit{rnd} = \mathit{\RNN}$ or when $a+b$ is not exactly
representable. Addressing the former case entails setting $a
\oplus_{\mathit{rnd}} b$ to be $+0$ whenever $a$ is $-b$. In the
latter case, $a \oplus_{\mathit{\RNZ}} b$ and $a \oplus_{\mathit{rnd}}
b$ deviate whenever $a \oplus_{\mathit{\RNZ}} b = \mathit{pred}(a+b)$
and $a \oplus_{\mathit{rnd}} b = \mathit{succ}(a+b)$, or vice versa.
The value $a \oplus_{\mathit{rnd}}
b$ is equal to $\mathit{succ}(a+b)$ when $a \oplus_{\mathit{RZ}} b =
\mathit{pred}(a+b)$ if $a+b \not \in \mathbb{T}$ and $0 < a+b$, while
the opposite scenario can occur when $a+b \not \in \mathbb{T}$ and
$a+b < 0$. Producing $a \oplus_{\mathit{RZ}} b$ thus entails adjusting
$a \oplus_{\mathit{rnd}} b$ whenever $a+b$ is not exactly
representable (\ie, $a+b \neq a \oplus_{\mathit{rnd}} b$) and $|a+b| <
|a \oplus_{\mathit{rnd}} b|$.

Based on the definitions of $\mathit{pred}$ and $\mathit{succ}$, the
absolute difference between $a \oplus_{\mathit{rnd}} b$ and $a
\oplus_{\mathit{\RNZ}} b$ when $|a+b| < |a \oplus_{\mathit{rnd}} b|$
is 1 ULP, which is the distance between two adjacent FP numbers around
$a+b$. We use the definition of ULP as described by
Overton~\cite{overton:fp:2001}: $\forall{t} \in \mathbb{T}, \mathit{ulp}(t) =
2^{e-p+1}$, where $e$ represents the exponent of $t$ and $p$
represents the available precision (\ie, $p=53$ for 64-bit doubles).
We note that because $a \oplus_{\mathit{rnd}} b$ and $a \oplus_{\mathit{\RNZ}} b$ are
faithful roundings of $a+b$, $\mathit{pred}(|a \oplus_{\mathit{rnd}} b|) = |a
\oplus_{\mathit{\RNZ}} b|$ when $|a+b| < |a \oplus_{\mathit{rnd}} b|$. Using Overton's
definition of $\mathit{ulp}(t)$, we can thus determine that when $|a+b| < |a
\oplus_{\mathit{rnd}} b|$, $|(a \oplus_{\mathit{rnd}} b) - (a \oplus_{\mathit{\RNZ}} b)| =
\mathit{ulp}(\mathit{pred}(|a \oplus_{\mathit{rnd}} b|))$. Given the definition of $\mathit{ulp}$ and the
relationship between $a \oplus_{\mathit{\RNZ}} b$ and $a \oplus_{\mathit{rnd}} b$, we
define the function $\mathit{\RZA}(a,b, a \oplus_{\mathit{rnd}} b)$ that returns $a
\oplus_{RZ} b$ as follows.

\setcounter{equation}{7}

\begin{equation}\label{get-rz-add}
\mathit{\RZA}(a, b, a \oplus_{\mathit{rnd}} b) = \begin{cases}
      +0, & \text{if $a$ is  $-b$} \\
      a \oplus_{\mathit{rnd}} b - (-1)^{\mathit{sign}(a \oplus_{\mathit{rnd}} b)} \times \mathit{ulp}(\mathit{pred}(|a \oplus_{\mathit{rnd}} b|)), & \text{if $|a+b| < |a \oplus_{\mathit{rnd}} b|$}\\
      a \oplus_{\mathit{rnd}} b & \text{else}
    \end{cases}
\end{equation}

The equation above is a specification for obtaining the
$\mathit{\RNZ}$ result of FP addition given the operands and the
original $a \oplus_{\mathit{rnd}} b$. However, the real arithmetic result of
$a+b$ is not directly computable in FP arithmetic. Hence, we design
Algorithm~\ref{alg:rz_add} to implement the specification in
Equation~\ref{get-rz-add}.

\begin{small}
\begin{algorithm}
  \caption{\small Our high level algorithm for computing the $\mathit{\RNZ}$
    result for addition (\ie, $a \oplus_{\mathit{\RNZ}} b$) given any rounding
    mode $\mathit{\mathit{rnd}}$. All FP operations are performed using
    double-precision FP arithmetic. Hence, $s = a
    \oplus_{\mathit{rnd}} b$, $z = (a \oplus_{\mathit{rnd}} b)
    \oplus_{\mathit{rnd}} (-a)$, and $t = b \oplus_{\mathit{rnd}} (-((a \oplus_{\mathit{rnd}} b)
    \oplus_{\mathit{rnd}} (-a)))$. Here, $\mathit{bit}(t)$ returns the IEEE-754 64-bit
    bit-string, and $\mathit{sign}(s)$ returns $0$ if $s$ is positive and $1$
    otherwise. We compute $s = s - (-1)^{\mathit{sign}(s)} \times
    \mathit{ulp}(\mathit{pred}(|s|))$ using bitwise operations, specifically by
    decrementing $\mathit{bit}(s)$ and then converting the resulting
    bit-pattern into a 64-bit FP number.}
    \label{alg:rz_add}    
    \SetKwProg{generate}{Function \emph{RZA(double a, double b)}}{}{end}
    \generate{}{
      \If {$\mathit{bit}(a) \ \textbf{xor} \ \mathit{bit}(b) == 0x8000000000000000$} {
        return $+0.0$;
      } 
      $double \ s = a+b$;\\
      \If {$|b| > |a|$}{
        $a, b = b, a$;
      }
      $double \ z = s - a$;  \\
      $double \ t = b - z$; \\
      \If {($\mathit{bit}(t)<<1 \neq 0$) \ \textbf{and} \  ($\mathit{bit}(t) \  \textbf{xor} \  bit(s) \geq 0x8000000000000000$)}{
        $s = s - (-1)^{\mathit{sign}(s)} \times \mathit{ulp}(\mathit{pred}(|s|))$;
      }
      return $s$;\\ 
    }
\end{algorithm}
\end{small}

\textbf{Emulating the $\mathit{\RNZ}$ result for addition.}
Algorithm~\ref{alg:rz_add} describes the procedure to compute $a
\oplus_{\mathit{\RNZ}} b$ given two 64-bit non-$\mathit{NaN}$,
non-infinity FP operands $a$ and $b$ and the rounding mode
$rnd$. Using lines 2 through 4, the algorithm applies the first case
of Equation~\ref{get-rz-add} and handles the cases where $a
\oplus_{\mathit{rnd}} b \neq a \oplus_{\mathit{\RNZ}} b$ because $a$
is $-b$ and $\mathit{rnd} = \mathit{\RNN}$. The expressions
$\mathit{bit}(a)$ and $\mathit{bit}(b)$ in line 2 provide the
bit-strings of the numbers $a$ and $b$ in the 64-bit IEEE-754
representation as 64-bit integers. The bit-strings of negative numbers
are greater than or equal to $0x8000000000000000$ (\ie, the sign bit
is 1). When two terms have different signs (\ie, $\mathit{sign}(a) =
1$ and $\mathit{sign}(b) = 0$, or vice versa), the sign bit of the
\textbf{xor} result will be set to $1$. The condition $\mathit{bit}(a)
\ \textbf{xor} \ \mathit{bit}(b) == 0x8000000000000000$ indicates that
$a$ and $b$ have the same absolute value but with different signs, the
exact circumstance under which $a \oplus_{\mathit{RZ}} b$ should
always return $+0$. The remaining steps of the algorithm address the
second case of Equation~\ref{get-rz-add}, in which $a+b$ is not
exactly representable and $|a+b| < |a \oplus_{\mathit{rnd}} b|$. These
steps are based on the following two theorems, which we prove in the
supplemental materials along with a detailed proof of $\mathit{\RZA}$.

\begin{theorem}\label{theorem:rz-add-1}
Let $a$ and $b$ be two non-$\mathit{NaN}$, non-infinity floating-point
numbers such that $a \oplus_{\mathit{rnd}} b$ does not overflow for
any rounding mode. If $a+b - (a \oplus_{\mathit{rnd}} b) \neq 0$, $a
\oplus_{\mathit{rnd}} b$ and $a+b - (a \oplus_{\mathit{rnd}} b)$ have
different signs if and only if $|a + b| < |a \oplus_{\mathit{rnd}}
b|$.
\end{theorem}

From Theorem~\ref{theorem:rz-add-1} we conclude that testing whether
$a+b - (a \oplus_{\mathit{rnd}} b)$ is a non-zero value, which would
indicate $a+b \neq a \oplus_{\mathit{rnd}} b$, and whether its sign
differs from that of $a \oplus_{\mathit{rnd}} b$ is sufficient for
determining if $a \oplus_{\mathit{rnd}} b$ needs to be adjusted
according to Equation~\ref{get-rz-add}. The focal point of
Algorithm~\ref{alg:rz_add} is checking the condition $a+b \neq a
\oplus_{rnd} b$ without having direct access to $a+b$. Lines 5 through
10 in Algorithm \ref{alg:rz_add} employ the steps in Dekker's
$FastTwoSum$ algorithm~\cite{dekker:fast2sum:1971} to compute the
value $t$, which is an approximation of the rounding error of the
initial FP addition (\ie $a+b - (a \oplus_{\mathit{rnd}} b)$). Since
all operations are performed in FP arithmetic, the outputs of the
subtractions in line 9 and 10 are subject to rounding. As a result,
$z=s \oplus_{\mathit{rnd}} (-a)= (a \oplus_{\mathit{rnd}} b)
\oplus_{\mathit{rnd}} (-a)$ and $t=b \oplus_{\mathit{rnd}} (-z) = b
\oplus_{\mathit{rnd}} (- ((a \oplus_{\mathit{rnd}} b)
\oplus_{\mathit{rnd}} (-a)))$. Proving the viability of $\RZA$ thus
requires affirming that $t$ is an appropriate proxy for $a+b - (a
\oplus_{\mathit{rnd}} b)$ with respect to applying
Theorem~\ref{theorem:rz-add-1} under all possible modes of
$\mathit{rnd}$. Specifically, $t$ must be sufficient for the purposes
of assessing whether $a+b - (a \oplus_{\mathit{rnd}} b) \neq 0$ and
$\mathit{sign}(a+b - (a \oplus_{\mathit{rnd}} b)) \neq \mathit{sign}(a
\oplus_{\mathit{rnd}} b)$.

An important point to note is that $t$ exactly represents the error
$a+b - (a \oplus_{\mathit{rnd}} b)$ when the rounding mode is
$\mathit{\RNE}$ and is a faithfully rounded FP value of the error $a +
b - (a \oplus_{\mathit{rnd}} b)$ for other rounding
modes~\cite{boldo:fast2sum:2017}. The suitability of $t$ as an
approximation of $a+b - (a \oplus_{\mathit{rnd}} b)$ follows directly
from Theorem~\ref{theorem:rz-add-2}.

\begin{theorem}\label{theorem:rz-add-2}
Let $t \in \mathbb{T}$ be a faithful rounding of the error $a+b - (a
\oplus_{\mathit{rnd}} b) \in \mathbb{R}$. Then, $t$ is neither $+0$
nor $-0$ if and only if $a+b - (a \oplus_{\mathit{rnd}} b) \neq 0$.
\end{theorem}

Given that $t$ is a faithful rounding of the real error $a+b - (a
\oplus_{\mathit{rnd}} b)$~(see our proofs in the supplemental material
and~\cite{boldo:fast2sum:2017}), Theorem \ref{theorem:rz-add-2}
asserts that the comparison $(\mathit{bit}(t)<<1) \neq 0$, which
checks if $t$ is neither $+0$ nor $-0$, is sufficient for determining
if $a+b - (a \oplus_{\mathit{rnd}} b) \neq 0$. Given
Theorem~\ref{theorem:rz-add-2}, $a+b - (a \oplus_{\mathit{rnd}} b)
\neq 0$ must hold when $(\mathit{bit}(t)<<1) \neq 0$ is
true. Subsequently, $(\mathit{bit}(t)<<1) \neq 0$ implies $t$ is a
faithful rounding of a non-zero real number, and thus the equality $t
= \mathit{rnd}(a+b - (a \oplus_{\mathit{rnd}} b))$ holds for all
rounding functions $\mathit{rnd} \in \{ \mathit{\RNE}, \mathit{\RNZ},
\mathit{\RNN}, \mathit{\RNP}\}$. Hence, one can apply
Lemma~\ref{lemma:sign-preservation}, which states the preservation of
sign with faithful rounding for non-zero real numbers, to conclude
that $(\mathit{bit}(t)<<1) \neq 0$ implies $\mathit{sign}(t) =
\mathit{sign}(a+b - (a \oplus_{\mathit{rnd}} b))$. The expression
$(\mathit{bit}(t) \ \textbf{xor} \ \mathit{bit}(s)) \geq
0x8000000000000000$, which checks if $\mathit{sign}(t) \neq
\mathit{sign}(s=a \oplus_{\mathit{rnd}} b)$, can thus accurately
assess if $\mathit{sign}(a+b - (a \oplus_{\mathit{rnd}} b)) \neq
\mathit{sign}(a \oplus_{\mathit{rnd}} b)$ when $(\mathit{bit}(t)<<1)
\neq 0$. Therefore, the conditions in line 11 of
Algorithm~\ref{alg:rz_add} associated with $t$ are sufficient for the
purposes of determining if $a+b - (a \oplus_{\mathit{rnd}} b) \neq 0$
and $\mathit{sign}(a+b - (a \oplus_{\mathit{rnd}} b)) \neq
\mathit{sign}(a \oplus_{\mathit{rnd}} b)$. When such conditions are
met, our algorithm $\RZA$ modifies the value of $a
\oplus_{\mathit{rnd}} b$ through line 12. In summary, our algorithm
identifies the presence of rounding error, determines whether the
error indicates $|a+b| < |a \oplus_{\mathit{rnd}} b|$, and accordingly
adjusts $a \oplus_{\mathit{rnd}} b$ to match $a \oplus_{\mathit{\RNZ}}
b$. Figure~\ref{fig:rz_add_example} provides an example of how our
algorithm $\mathit{\RZA}$ adjusts a non-$\mathit{\RNZ}$ FP addition
result based on Algorithm~\ref{alg:rz_add}.

\textbf{Relationship between $\otimes_{\mathit{\RNZ}}$ and other
  $\otimes_{\mathit{rnd}}$.}
Our objective for handling $a \otimes_{\mathit{rnd}} b$ is to adjust
its value whenever it is not equal to $a \otimes_{\mathit{\RNZ}} b$. A
notable difference between FP addition and multiplication is that the
latter operation does not entail any corner cases involving $+0$ or
$-0$ - whenever $a \times b = 0$, $a \otimes_{\mathit{rnd}} b$
displays the same behavior across all $\mathit{rnd}$. Due to this
distinction, the only cases of concern are those in which $a \times b$
is not exactly representable. As detailed in our premise for the
algorithm $\mathit{\RZA}$, $a \otimes_{\mathit{rnd}} b$ can differ
from $a \otimes_{\mathit{\RNZ}} b$ when the two numbers are distinct
FP neighbors of a product $a \times b$ that is not exactly
representable.  Because $a \otimes_{\mathit{\RNZ}} b$ should always
return the FP neighbor of $a \times b$ with the smaller absolute
value, $a \otimes_{\mathit{rnd}} b$ must be the FP neighbor with the
larger absolute value for $a \otimes_{\mathit{\RNZ}} b \neq a
\otimes_{\mathit{rnd}} b$ to hold. Given how a real number's FP
neighbors are defined in Equations~1 and~2 (Figure
~\ref{fig:background:modes}), $a \otimes_{\mathit{\RNZ}} b \neq a
\otimes_{\mathit{rnd}} b$ implies $|a \otimes_{\mathit{\RNZ}} b| = |a
\otimes_{\mathit{rnd}} b| - \mathit{ulp}(\mathit{pred}(|a
\otimes_{\mathit{rnd}} b|))$.  We define the specification of
$\mathit{\RZM}(a, b, a \otimes_{\mathit{rnd}} b)$ that returns $a
\otimes_{\mathit{\RNZ}} b$ under all rounding modes as follows.

\begin{equation}\label{get-rz-mult}
\mathit{\RZM}(a, b, a \otimes_{\mathit{rnd}} b) = \begin{cases}
      a \otimes_{\mathit{rnd}} b - (-1)^{\mathit{sign}(a \otimes_{\mathit{rnd}} b)} \times \mathit{ulp}(\mathit{pred}(|a \otimes_{\mathit{rnd}} b|)), & \text{if $|a \times b| < |a \otimes_{\mathit{rnd}} b|$}\\
      a \otimes_{\mathit{rnd}} b & \text{else}
    \end{cases}
\end{equation}

\begin{small}
\begin{algorithm}
  \caption{\small Our algorithm for computing $a
    \otimes_{\mathit{\RNZ}} b$ given any rounding mode
    $\mathit{rnd}$. All FP operations are performed using
    double-precision FP arithmetic. Hence, $m = a
    \otimes_{\mathit{rnd}} b$, $c1 = \mathit{fma}_{\mathit{rnd}}(a,b,
    -m)$, and $c2 = \mathit{fma}_{\mathit{rnd}}(-a,b, m)$. Given
    non-$\mathit{NaN}$, non-infinity operands $a$, $b$, and $c$, the
    fused-multiply-add instruction
    $\mathit{fma}_{\mathit{rnd}}(a,b,c)$ returns a faithful rounding
    of $a \times b + c$. Here, $c1$ is a faithful rounding of $a
    \times b - m$ and $c2$ is a faithful rounding of $(-a) \times b +
    m$. The remaining details are analogous to those found in
    Algorithm~\ref{alg:rz_add}.}
    \label{alg:rz_mult}    
    \SetKwProg{generate}{Function \emph{\RZM(double a, double b)}}{}{end}
    \generate{}{
      $double \ m = a \times b$; \\
      $double \ c1 = \mathit{fma}(a, b, -m)$;\\
      $double \ c2 = \mathit{fma}(-a, b, m)$;\\
      \If {$\mathit{bit}(c1) \neq \mathit{bit}(c2)$ \ \textbf{and} \ $(\mathit{bit}(c1) \ \textbf{xor} \ \ \mathit{bit}(m)) \geq 0x8000000000000000$}{
        $m = m - (-1)^{\mathit{sign}(m)} \times \mathit{ulp}(\mathit{pred}(|m|))$;
      }
      return $m$;\\ 
    }
\end{algorithm}
\end{small}

\textbf{Emulating the $\mathit{\RNZ}$ result for multiplication.} As
is the case with FP addition, simulating the $\mathit{\RNZ}$ result
for FP multiplication is non-trivial because the real value output $a
\times b$ is not directly observable. Algorithm~\ref{alg:rz_mult}
describes our approach to computing $a \otimes_{\mathit{\RNZ}} b$
given two non-$\mathit{NaN}$, non-infinity 64-bit FP numbers $a$ and
$b$ for which the product doesn't cause overflow. Much like
Algorithm~\ref{alg:rz_add}, it computes a faithful rounding of the
error in the original FP operation to assess whether the original
product needs to be adjusted to match the $\mathit{\RNZ}$ result. As
we describe in our proofs for $\mathit{\RZM}$, however, a faithful
rounding of $a \times b - (a \otimes_{\mathit{rnd}} b)$ could be equal
to $+0$ or $-0$ even when $a \times b - (a \otimes_{\mathit{rnd}} b)
\neq 0$ because rounding error induced by multiplication is
susceptible to underflow. We therefore rely on the following two
theorems to confirm the presence and nature of the rounding error in
$a \otimes_{\mathit{rnd}} b$.

\begin{theorem}
\label{theorem:rz-mult-1}
Let $a$ and $b$ be two non-$\mathit{NaN}$, non-infinity floating-point
numbers such that $a \otimes_{\mathit{rnd}} b$ does not overflow for
any rounding mode. If $a \times b - (a \otimes_{\mathit{rnd}} b) \neq
0$, $a \otimes_{\mathit{rnd}} b$ and $a \times b - (a
\otimes_{\mathit{rnd}} b)$ have different signs if and only if $|a
\times b| < |a \otimes_{\mathit{rnd}} b|$.
\end{theorem}

\begin{theorem}
\label{theorem:rz-mult-2}
Let $a$ and $b$ be two non-$\mathit{NaN}$, non-infinity floating-point
numbers such that $a \otimes_{\mathit{rnd}} b$ does not overflow for any
rounding mode. Let $\mathit{bit}(f)$ be a function that returns the bit-string
of any floating-point number $f$. Then, for any rounding mode $\mathit{rnd}$,
$\mathit{bit}(\mathit{fma}_{\mathit{rnd}}(a, b, -(a \otimes_{\mathit{rnd}} b))
\neq \mathit{bit}(\mathit{fma}_{\mathit{rnd}}(-a, b, a
\otimes_{\mathit{rnd}} b))$ if and only if $a \times b - (a
\otimes_{\mathit{rnd}} b) \neq 0$.
\end{theorem}

Lines~2 and 3 in Algorithm \ref{alg:rz_mult} compute $m = a
\otimes_{rnd} b$ and $c1 = \mathit{fma}_{\mathit{rnd}}(a, b, -(a
\otimes_{\mathit{rnd}} b))$.  The output of a fused-multiply-add (FMA)
operation in the form of $\mathit{fma}_{\mathit{rnd}}(a, b, c)$ is the
result of performing a faithfully rounded FP addition
$\oplus_{\mathit{rnd}}$ using the operands $a \times b$ and $c$ with
no intermediate rounding for the multiplication. Based on the rules of
$\oplus_{\mathit{rnd}}$ detailed in Equation~6
(Figure~\ref{fig:background:ops}), $c1 =
\mathit{fma}_{\mathit{rnd}}(a, b, -m)$ is guaranteed to be a faithful
rounding of $a \times b - m = a \times b - (a \otimes_{\mathit{rnd}}
b)$. Line~4 of the algorithm computes $c2 =
\mathit{fma}_{\mathit{rnd}}(-a, b, m)$, which makes $c2$ a faithful
rounding of $(-a) \times b +m = -(a \times b) + (a \otimes_{\mathit{rnd}}
b)$. The algorithm computes $c2$ to compare its bit-pattern against
that of $c1$ to utilize Theorem~\ref{theorem:rz-mult-2}, which
guarantees that $a \times b - (a \otimes_{\mathit{rnd}} b) \neq 0$ whenever the
condition $\mathit{bit}(c1) \neq \mathit{bit}(c2)$ in line 5 is true.

Once it is established that $a \times b - (a \otimes_{\mathit{rnd}} b)
\neq 0$, Theorem~\ref{theorem:rz-mult-1} affirms that examining the
sign of $a \times b - (a \otimes_{\mathit{rnd}} b)$ relative to $a
\otimes_{\mathit{rnd}} b$ is sufficient for assessing whether $|a
\times b| < |a \otimes_{\mathit{rnd}} b|$. Here, $c1$ is a faithful
rounding of a non-zero real value $a \times b - (a
\otimes_{\mathit{rnd}} b)$. From Lemma~\ref{lemma:sign-preservation},
the sign preserving properties of faithful rounding with respect to
non-zero real numbers ensures that $\mathit{sign}(c1) =
\mathit{sign}(a \times b - (a \otimes_{\mathit{rnd}} b))$. Under such
conditions, $\mathit{\RZM}$ can thus apply Theorem
~\ref{theorem:rz-mult-1} and use $c1$ as a proxy for $a\times b - (a
\otimes_{\mathit{rnd}} b)$. Our algorithm $\mathit{\RZM}$ applies
Theorem ~\ref{theorem:rz-mult-1} by confirming $a \times b - (a
\otimes_{\mathit{rnd}} b) \neq 0$ through the condition
$\mathit{bit}(c1) \neq \mathit{bit}(c2)$ and checking if
$\mathit{sign}(c1) \neq \mathit{sign}(m=a \otimes_{\mathit{rnd}} b)$
through the expression $(\mathit{bit}(c1) \ \textbf{xor}
\ \ \mathit{bit}(m)) \geq 0x8000000000000000$. In conclusion, Theorems
\ref{theorem:rz-mult-1} and $\ref{theorem:rz-mult-2}$ in conjunction
with Lemma~\ref{lemma:sign-preservation} guarantee that the conditions
$\mathit{bit}(c1) \neq \mathit{bit}(c2)$ and $(\mathit{bit}(c1)
\ \textbf{xor} \ \ \mathit{bit}(m)) \geq 0x8000000000000000$ are
appropriate for testing whether $a \times b - (a
\otimes_{\mathit{rnd}} b) \neq 0$, $\mathit{sign}(a \times b - (a
\otimes_{\mathit{rnd}} b)) \neq \mathit{sign}(a \otimes_{\mathit{rnd}}
b)$, and $|a \times b| < |a \otimes_{\mathit{rnd}} b|$. When such
conditions are met, line~6 of Algorithm~\ref{alg:rz_mult} produces $a
\otimes_{\mathit{\RNZ}} b$ by adjusting the value of $a
\otimes_{\mathit{rnd}} b$ in accordance with Equation
~\ref{get-rz-mult}. We provide detailed proofs for the algorithm
$\mathit{\RZM}$ and its associated theorems in the supplemental
materials (see Section~7). Figure ~\ref{fig:rz_mult_example} provides
an example of how $\mathit{\RZM}$ adjusts a non-\RNZ FP multiplication
result based on Algorithm~\ref{alg:rz_mult}.

\begin{figure}
  \begin{subfigure}{0.34\linewidth}
  \includegraphics[width=\linewidth, height=3.5cm]{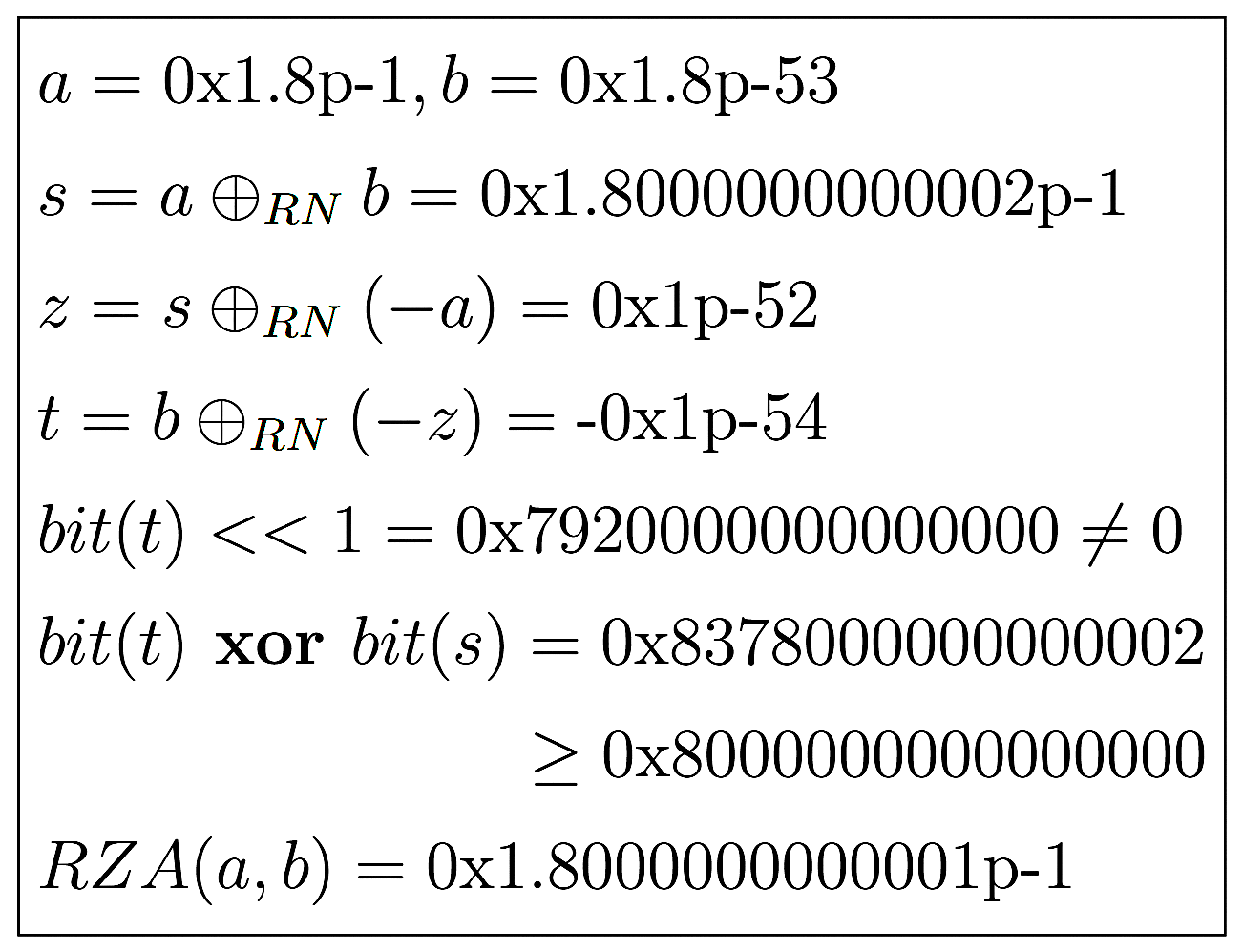}
  \caption{}\label{fig:rz_add_example}
  \end{subfigure}
  \hspace{15mm}
  \begin{subfigure}{0.34\linewidth}
  \includegraphics[width=\linewidth, height=3.5cm]{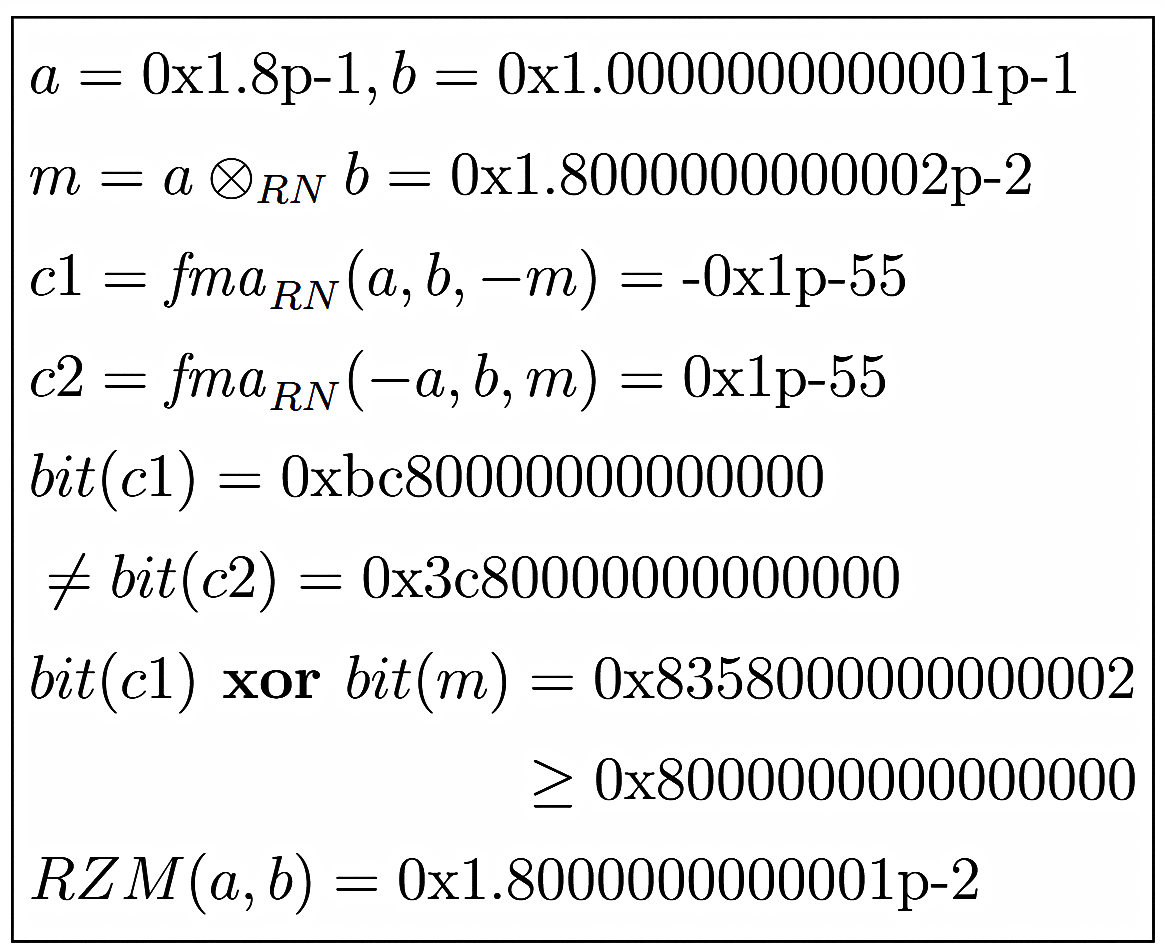}
  \caption{}\label{fig:rz_mult_example}
  \end{subfigure}
  \caption{\small (a) An illustration of how our $\mathit{\RZA}$
    function adjusts the output of a double-precision FP addition
    result $s = a \oplus_{\mathit{rnd}} b$ when $s \neq a
    \oplus_{\mathit{\RNZ}} b$. (b) An illustration of how our
    $\mathit{\RZM}$ function adjusts the output of a double-precision
    FP multiplication result $m = a \otimes_{\mathit{rnd}} b$ when $m
    \neq a \otimes_{\mathit{\RNZ}} b$. All FP values are shown as hex
    floats. All integers (\eg, $\mathit{bit}(t)$) are shown in
    hexadecimal representation.}
    \label{fig:rz_emulation}    
\end{figure}

\textbf{Changes to the \rlibm pipeline.} The rounding-invariant
outputs approach requires minimal changes to the \rlibm pipeline. The
only change needed is to set the rounding mode to $\mathit{\RNZ}$
instead of $\mathit{\RNE}$ when performing range reduction, reduced
interval generation, and polynomial generation. For the resultant
implementations, this new method entails replacing all FP additions
and multiplications, which are the only FP operations in the code that
induce rounding-related variability, with $\mathit{\RZA}$ and
$\mathit{\RZM}$, respectively. These modifications lead to
implementations that can produce correctly rounded results for all
inputs across multiple representations and rounding modes without
requiring explicit changes to the application-level rounding mode.

\subsection{Rounding Independence by Deducing Rounding-Invariant Input Bounds}
\label{sec:multi-round}

The underlying principle behind our round-to-zero emulation in
Section~\ref{subsec:round-to-zero} is the internal enforcement of a
single rounding mode (\ie, \RNZ) such that the range reduction,
polynomial evaluation, and output compensation steps in a given \rlibm
implementation will always produce the same outputs regardless of the
invoking application's rounding mode. This approach requires
augmenting each FP addition and multiplication with operations
designed to produce the \RNZ result, which can add some overhead to
the final implementations. To address this issue, we propose an
alternative method that bypasses the overheads required for
accomplishing \textbf{rounding mode-invariant outputs} by deducing the
\textbf{rounding mode-invariant bounds} on the polynomial evaluation
and output compensation results for the reduced inputs
encountered. This method moves the overhead from the final math
library implementation to the process of generating it.

The range reduction, polynomial evaluation, and output compensation
processes involve FP arithmetic and are potentially sensitive to
rounding modes. Among them, range reduction is the least sensitive
because the algorithms used generally produce rounding
mode-independent results.
We make all range reduction algorithms used with our new approach
produce results that conform to $\mathit{\RNZ}$. We enforce this
requirement in the resultant implementations by replacing all rounding
mode-dependent FP additions and multiplications used for range
reduction with the functions $\mathit{\RZA}$ and $\mathit{\RZM}$
discussed in Section~\ref{subsec:round-to-zero}.

\begin{figure}
  \begin{subfigure}{0.44\linewidth}
  \includegraphics[width=\linewidth]{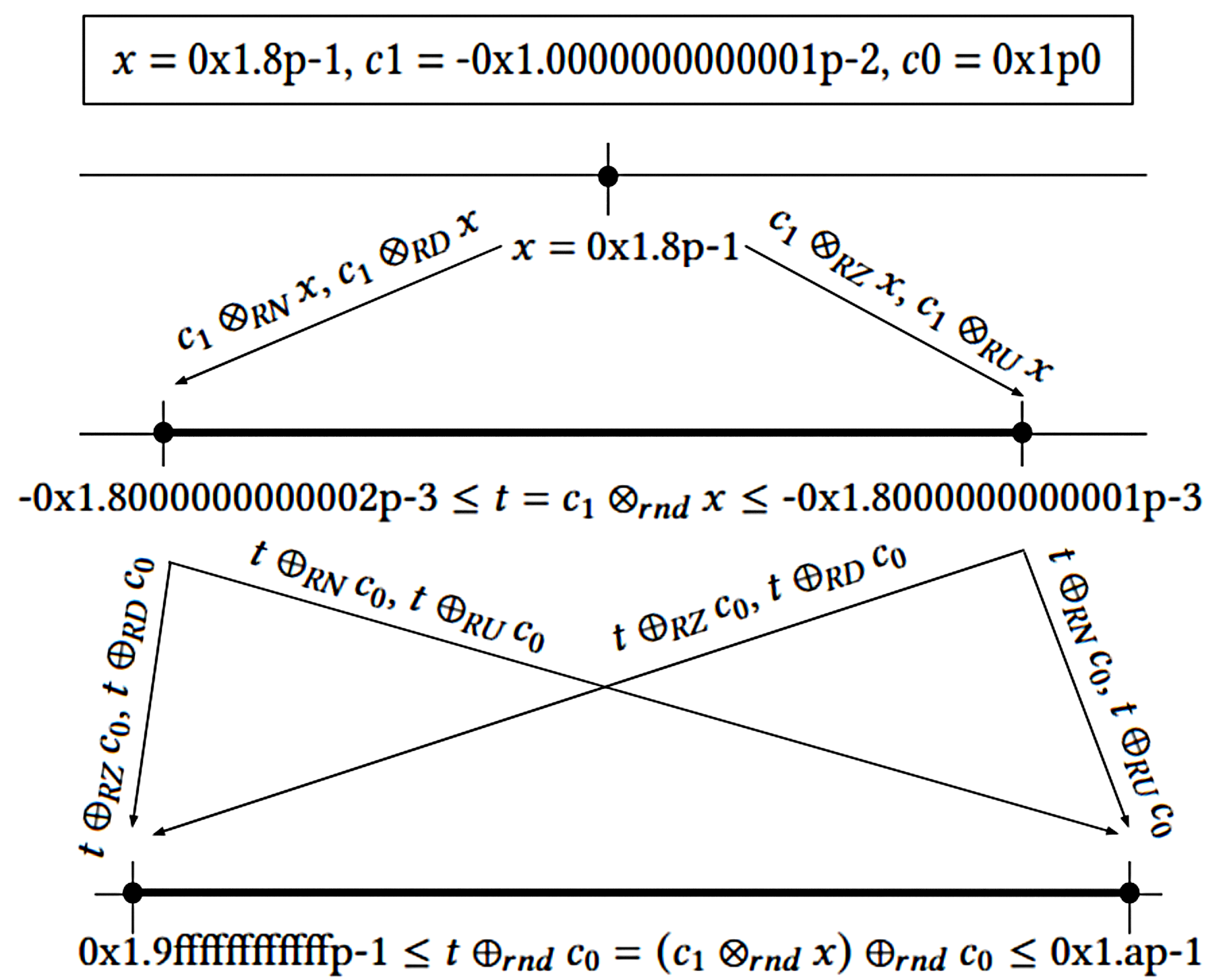}
  \caption{}\label{fig:input_bounds_intuition}
  \end{subfigure}
  \hspace{15mm}
  \begin{subfigure}{0.22\linewidth}
  \includegraphics[width=\linewidth, height=4.5cm]{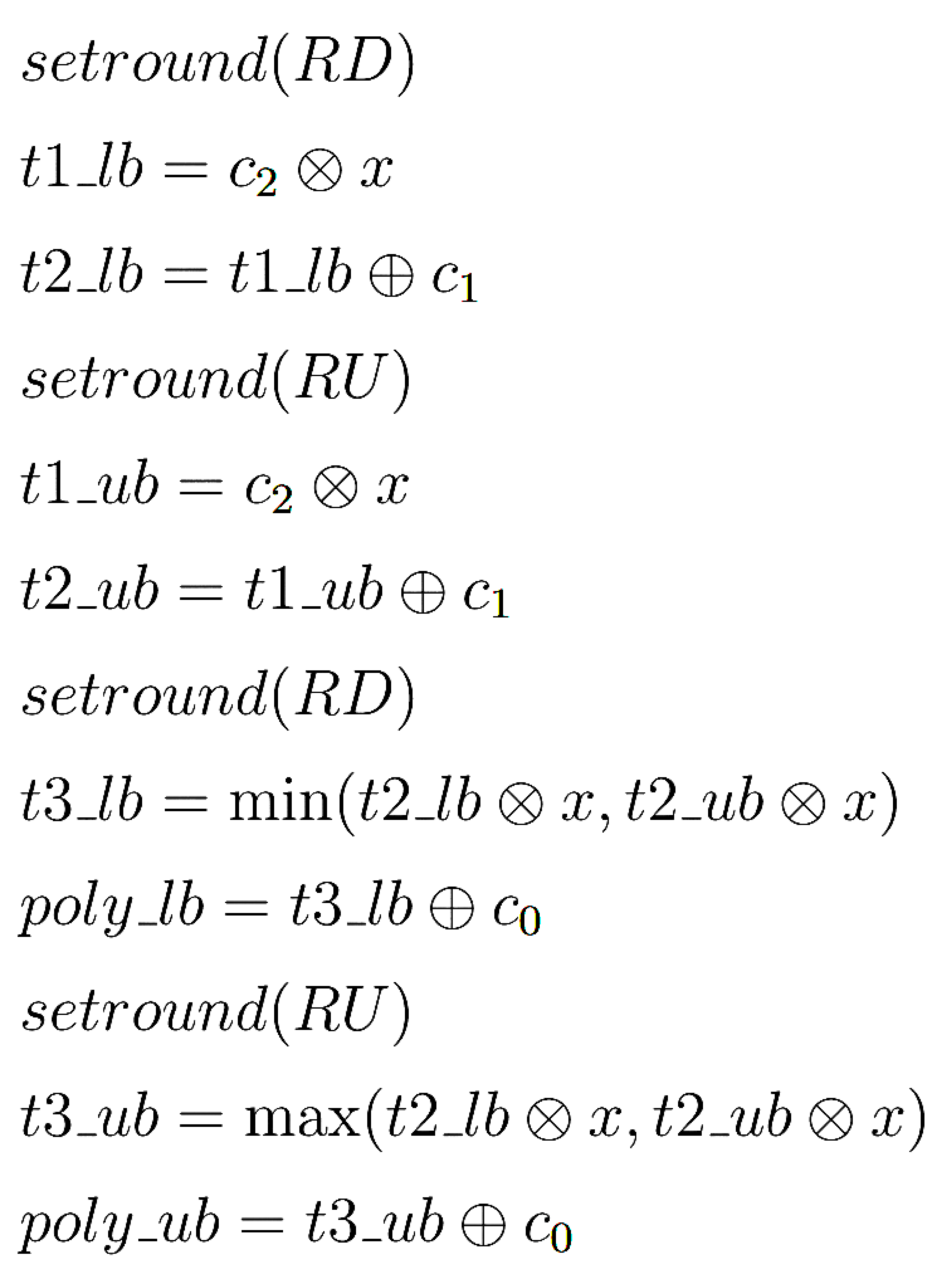}
  \caption{}\label{fig:input_bounds_example}
  \end{subfigure}
  \caption{\small (a) The figure illustrates the intuition behind how
    our rounding-invariant input bounds approach identifies the lower
    and upper bounds of a polynomial output across all rounding modes
    using a polynomial of the form $P(x) = (c_{1} \otimes_{rnd} x)
    \oplus_{rnd} c_{0}$ as an example. We highlight that the lower
    bounds for both the $\otimes_{rnd}$ and $\oplus_{rnd}$ operations
    involve the \RNN results. Similarly, the upper bounds involve the
    \RNP results.  (b) The figure illustrates the steps the
    rounding-invariant input bounds approach takes to compute the
    lower bound ($poly\_lb$) and upper bound ($poly\_ub$) of a
    polynomial $P(x) = (((c_{2} \otimes_{rnd} x)\oplus_{rnd} c_{1})
    \otimes_{rnd} x) \oplus_{rnd} c_{0}$ during the library generation
    process.}
    \label{fig:input_bounds}    
\end{figure}

\textbf{Accounting for rounding mode induced variability.} 
To account for the variability in the FP evaluation of polynomials and
output compensation functions under different rounding modes, we
deduce new bounds for the reduced intervals given to the polynomial
generator and for the expected results of the polynomial
approximations. Our key idea is to consider the polynomial
approximations and output compensation functions as returning not a
single value, but rather a \textbf{range of values}.  Representing the
output of a sequence of FP operations as a range of values (\ie, an
interval) accounts for the variability in the results stemming from
the different rounding modes.  We use interval arithmetic to identify
bounds for such ranges, which the new pipeline uses during the reduced
interval and polynomial generation stages. We construct this new
strategy based on the insight that the round-down (\RNN) and round-up
(\RNP) results serve as the lower and upper bounds for a given FP
operation's output range across various rounding
modes. Figure~\ref{fig:input_bounds_intuition} illustrates the
intuition behind this method. We now describe our approach for
deducing the bounds of polynomial evaluation and output compensation
results in the context of our reduced interval and polynomial
generation processes.

\textbf{Accounting for rounding induced variability in polynomial evaluation.}
Let $P(x) = \sum_{i=0}^{d}c_{i}x^{i} = c_{d}x^{d} + c_{d-1}x^{d-1} +
\cdot\cdot\cdot + c_{2}x^{2} + c_{1}x + c_{0}$ where $P(x) \in
\mathbb{R}$ and $\forall{i}, c_{i} \in \mathbb{T}$ represent the
typical polynomial approximation considered in the \rlibm project. If
a polynomial $P(x)$ has at least one non-constant term, the result of
evaluating $P(x)$ on a given input $x$ via real arithmetic can be
expressed as either $P(x)=P_{1}(x)\times P_{2}(x)$ or
$P(x)=P_{1}(x)+P_{2}(x)$ where $P_{1}(x), P_{2}(x) \in \mathbb{R}$
represent the appropriate intermediate terms. If the evaluation is
performed using FP arithmetic, the final output may differ depending
on the order in which operations are performed since FP operations are
non-associative. For our purposes, we assume throughout the paper that
every polynomial is associated with a unique, predetermined FP
evaluation scheme. The variability induced by the different orderings
that are possible when translating a polynomial to a sequence of FP
operations is outside the scope of this paper. Even when evaluating
$P(x)$ using a fixed order of FP operations, one may observe different
outputs depending on the rounding rule applied to each
operation. Henceforth, we use $P(x)$ to represent the final output
obtained from evaluating a given polynomial using $n$ ordered FP
operations each subject to a rounding rule $\mathit{rnd}_{i} \in \{
\mathit{\RNE}, \mathit{\RNZ}, \mathit{\RNN}, \mathit{\RNP}
\}$. Recognizing the impact of rounding on the final output, we
formulate $P(x)$ through the following definition.

\begin{small}
\begin{definition}
  \label{poly-eval-def}
\begin{equation}\nonumber
P(x) = \begin{cases} c, & \text{if $P(x)$ is a FP constant}\\
    P_{1}(x) \oplus_{\mathit{\RNE}} P_{2}(x) & \text{if $\mathit{flop}_{n}=\oplus$ and $\mathit{rnd}_{n} = \mathit{\RNE}$}\\
    P_{1}(x) \oplus_{\mathit{\RNZ}} P_{2}(x) & \text{if $\mathit{flop}_{n}=\oplus$ and $\mathit{rnd}_{n} = \mathit{\RNZ}$}\\
    \cdot \cdot \cdot \\
    P_{1}(x) \otimes_{\mathit{\RNN}} P_{2}(x) & \text{if $\mathit{flop}_{n}=\otimes$ and $\mathit{rnd}_{n} = \mathit{\RNN}$}\\
    P_{1}(x) \otimes_{\mathit{\RNP}} P_{2}(x) & \text{if $\mathit{flop}_{n}=\otimes$ and $\mathit{rnd}_{n} = \mathit{\RNP}$} 

      \end{cases}\\
\end{equation}
\end{definition}
\end{small}

The definition of $P(x)$ presented above is intended to be recursive:
$P_{1}(x)$ and $P_{2}(x)$ both represent the result of a specific
sequence of intermediate FP operations that are each subject to
different rounding rules. Here, $\mathit{rnd}_{n}$ represents the
rounding rule employed by the last FP addition or multiplication,
which we collectively denote as $\mathit{flop}_{n}$. The subscript in
$\mathit{rnd}_{n}$ indicates that each intermediate FP operation
$\mathit{flop}_{i}$ is executed under the corresponding rounding mode
$rnd_{i}$. Typically, the rounding rule used by each FP operation
would be fixed to the application's rounding mode. Note that for the
purposes of defining $P(x)$ in this paper, we do not require that all
instances of $\mathit{rnd}_{i}$ are identical.

Given a set of reduced inputs and their reduced intervals,
$(x_{i}',[l_{i}', h_{i}'])$, the goal of \rlibm's polynomial generator
is to find a polynomial with an FP output $P(x)$ as defined under
Definition~\ref{poly-eval-def} such that $l_{i}' \leq P(x_{i}') \leq
h_{i}'$ holds for all inputs $x_i$.  In the original \rlibm approach,
every intermediate FP output encountered while computing $P(x')$ for a
given reduced input is the result of rounding via \RNE (\ie,
$\forall{i}, \mathit{rnd}_{i} = \mathit{\RNE}$). In our rounding-invariant outputs
approach, every $\mathit{rnd}_{i} = \mathit{\RNZ}$. In both cases, $P(x_{i}')$ would be
a single value given the absence of rounding-induced variability in
each intermediate FP output. Figure~\ref{fig:poly-constraints}(a)
shows the constraints provided to the polynomial generator when the
output of the polynomial evaluation is a single value.

\textbf{Constraints for polynomial generation when outputs are treated
  as ranges of values.}
Given the same values for $P_{1}(x_{i}')$ and $P_{2}(x_{i}')$, the
final value of $P(x_{i}')$ may differ depending on the choice of
$rnd_{n}$. Thus, evaluating a single version of $P(x_{i}')$ derived
from a fixed $\mathit{rnd}_{n}$ provides limited information. For example, it
could be the case that the constraint $l_{i}' \leq P(x_{i}') \leq
h_{i}'$ holds for a reduced input $x_{i}'$ and its reduced interval
$[l_{i}', h_{i}']$ when $\mathit{rnd}_{n}= \mathit{\RNE}$ \ but not when $\mathit{rnd}_{n} =
\mathit{\RNP}$. The problem is further exacerbated by the fact that a
fully-evaluated $P(x_{i}')$ is the result of a specific sequence of
intermediate roundings $\mathit{rnd}_{i}$. Different combinations of the
intermediate $\mathit{rnd}_{i}$ instances prior to $\mathit{rnd}_{n}$ could lead to
different values for $P_{1}(x_{i}')$ and $P_{2}(x_{i}')$. To account
for this variability, we adopt a new viewpoint regarding
$P(x)$. \emph{We no longer consider the final polynomial evaluation
output $P(x)$ for a particular reduced input $x_{i}'$ to be a single
FP number resulting from a fixed sequence of $\mathit{rnd}_{i}$. Rather, we
treat it as a function of the $n$ variable instances of $\mathit{rnd}_{i}$, the
output range of which we denote $[\underline{P(x_{i}')},
  \overline{P(x_{i}')}]$}. In the case that $l_{i}' \leq
\underline{P(x_{i}')} \leq \overline{P(x_{i}')} \leq h_{i}'$, one can
conclude that regardless of the rounding rule $\mathit{rnd}_{i}$ used for each operation
$\mathit{flop}_{i}$, the final result of $P(x_{i}')$ will satisfy its
associated constraint. Therefore, the goal of the new polynomial
generator would be to find a polynomial such that $l_{i}' \leq
\underline{P(x_{i}')} \leq \overline{P(x_{i}')} \leq h_{i}'$ for as
many $x_{i}$'s as possible. Figure~\ref{fig:poly-constraints}(b)
displays how the polynomial generator would apply constraints to a
candidate polynomial approximation under the new approach.

\begin{figure}
  \includegraphics[width=\linewidth]{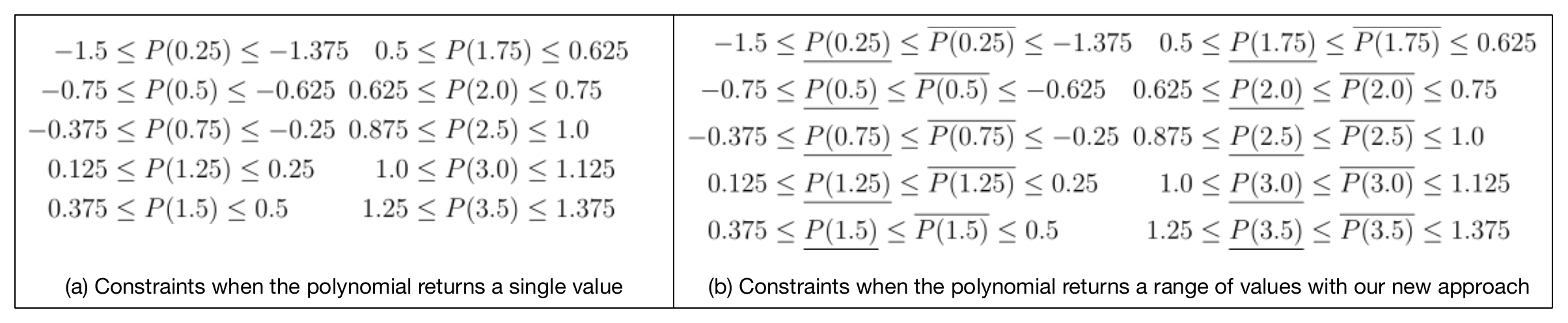}
  \caption{(a) Constraints used to produce the correctly rounded
    results when $P(x_{i}')$ is considered a single FP value as in the
    original \rlibm and our rounding-invariant outputs approaches. (b)
    Constraints used to produce the correctly rounded results when
    $P(x_{i}')$ is considered to be any FP value in the range
    $[\underline{P(x_{i}')}, \overline{P(x_{i}')}]$.}
    \label{fig:poly-constraints}
\end{figure}

\textbf{Defining the bounds of polynomial evaluation.} Computing
$\underline{P(x)}$ and $\overline{P(x)}$ for an arbitrary input $x$
necessitates identifying an appropriate rounding $rnd_{i}$ to apply to
each FP operation and identifying operands along with their respective
interval bounds.
We use the monotonicity properties of faithfully rounded FP operations
and identify the interval bounds using interval arithmetic rules
presented in Lemmas~\ref{lemma:int-sum} and~\ref{lemma:int-prod}.
Given the definition of $P(x)$ in Definition~\ref{poly-eval-def}, we
now present the definitions of $\underline{P(x)}$ and
$\overline{P(x)}$ through Theorems~\ref{theorem:lb-def} and
\ref{theorem:ub-def} and their proofs.

\begin{small}
\begin{theorem}\label{theorem:lb-def}
\begin{equation}\nonumber
\underline{P(x)} = \begin{cases}
      c, & \text{if FP constant}\\
      \underline{P_{1}(x)} \oplus_{\mathit{\RNN}} \underline{P_{2}(x)}, & \text{if $\mathit{flop}_{n}= \oplus$}\\
      \mathit{min}(
      \underline{P_{1}(x)} \otimes_{\mathit{\RNN}} \underline{P_{2}(x)}, 
      \underline{P_{1}(x)} \otimes_{\mathit{\RNN}} \overline{P_{2}(x)},  
      \overline{P_{1}(x)} \otimes_{\mathit{\RNN}} \underline{P_{2}(x)},
      \overline{P_{1}(x)} \otimes_{\mathit{\RNN}} \overline{P_{2}(x)}) & \text{if $\mathit{flop}_{n}= \otimes$}
     \end{cases}\\
\end{equation}
\end{theorem}
\end{small}

\begin{small}
\begin{theorem}\label{theorem:ub-def}
\begin{equation}\nonumber
\overline{P(x)} = \begin{cases}
      c, & \text{if FP constant}\\
      \overline{P_{1}(x)} \oplus_{\mathit{\RNP}} \overline{P_{2}(x)}, & \text{if $\mathit{flop}_{n}= \oplus$}\\
      \mathit{max}(
      \underline{P_{1}(x)} \otimes_{\mathit{\RNP}} \underline{P_{2}(x)}, 
      \underline{P_{1}(x)} \otimes_{\mathit{\RNP}} \overline{P_{2}(x)},  
      \overline{P_{1}(x)} \otimes_{\mathit{\RNP}} \underline{P_{2}(x)},
      \overline{P_{1}(x)} \otimes_{\mathit{\RNP}} \overline{P_{2}(x)}) & \text{if $\mathit{flop}_{n}= \otimes$}
     \end{cases}\\
\end{equation}
\end{theorem}
\end{small}
\begin{proof}
Given the symmetry between the definitions of $\underline{P(x)}$ and
$\overline{P(x)}$ presented above, we primarily focus on proving
Theorem ~\ref{theorem:lb-def}. The first case in which $P(x)$ is a FP
constant requires no proof as it does not involve any FP operations
and $P(x) = c$ in this case across any rounding mode. We thus move on
to proving the definition for $\underline{P(x)}$ when
$\mathit{flop}_{n}=\oplus$.  Let $\underline{P(x)} = p_{1}^{*}
\oplus_{\mathit{rnd}_{n}^{*}} p_{2}^{*}$ where $p_{1}^{*} \in [
  \underline{P_{1}(x)}, \overline{P_{1}(x)} ]$, $ p_{2}^{*} \in [
  \underline{P_{2}(x)}, \overline{P_{2}(x)} ]$, and
$\mathit{rnd}_{n}^{*} \in \{\mathit{\RNE}, \mathit{\RNZ},
\mathit{\RNN}, \mathit{\RNP} \}$. We prove the correctness of
Theorem~\ref{theorem:lb-def}'s definition for this case by affirming
that $\mathit{rnd}_{n}^{*} = \mathit{\RNN}$ and that this equality
implies $p_{1}^{*} = \underline{P_{1}(x)}$ and $p_{2}^{*} =
\underline{P_{2}(x)}$. Suppose that there exists some
$\mathit{rnd}_{n}^{*} \in \{ \mathit{\RNE}, \mathit{\RNZ},
\mathit{\RNP} \}$ such that $p_{1}^{*} \oplus_{\mathit{rnd}_{n}^{*}}
p_{2}^{*} < p_{1}^{*} \oplus_{\mathit{\RNN}} p_{2}^{*}$. Such an
assumption would directly violate the bounds of
$\oplus_{\mathit{rnd}}$ established in
Lemma~\ref{lemma:rnd-sum-bound}, thereby establishing
$\mathit{rnd}_{n}^{*} = \mathit{\RNN}$. Having established
$\mathit{rnd}_{n}^{*} = \mathit{\RNN}$, we also prove that
$\mathit{rnd}_{n}^{*} = \mathit{\RNN}$ implies $p_{1}^{*} =
\underline{P_{1}(x)}$ and $p_{2}^{*} = \underline{P_{2}(x)}$ via
contradiction. Suppose $\mathit{rnd}_{n}^{*} = \mathit{\RNN}$ and
there exist a $p_{1}^{*}$ and a $p_{2}^{*}$ such that $p_{1}^{*}
\oplus_{\mathit{\RNN}} p_{2}^{*} < \underline{P_{1}(x)}
\oplus_{\mathit{\RNN}} \underline{P_{2}(x)}$. By the rules of interval
addition detailed in Lemma~\ref{lemma:int-sum}, it must be the case
that $\underline{P_{1}(x)} + \underline{P_{2}(x)} \leq p_{1}^{*} +
p_{2}^{*}$. The lower bound of the real arithmetic sum indicates that
the existence of a pair of operands $p_{1}^{*}$ and $p_{2}^{*}$ such
that $p_{1}^{*} \oplus_{\RNN} p_{2}^{*} < \underline{P_{1}(x)}
\oplus_{\mathit{\RNN}} \underline{P_{2}(x)}$ would directly contradict
the monotonic properties of $\oplus_{\mathit{\RNN}}$ detailed in
Lemma~\ref{lemma:rnd-sum-monotonic}. The resulting contradiction
indicates that $\mathit{rnd}_{n}^{*} = \mathit{\RNN}$ implies
$\underline{P(x)} = \underline{P_{1}(x)} \oplus_{\mathit{\RNN}}
\underline{P_{2}(x)}$, thereby concluding our proof for the case in
which $\mathit{flop}_{n} = \oplus$.

Our proof for the case in which $\mathit{flop}_{n} = \otimes$ largely
follows the same structure. Let $\underline{P(x)} = p_{1}^{*}
\otimes_{\mathit{rnd}_{n}^{*}} p_{2}^{*}$ where $p_{1}^{*} \in [
  \underline{P_{1}(x)}, \overline{P_{1}(x)} ]$, $ p_{2}^{*} \in [
  \underline{P_{2}(x)}, \overline{P_{2}(x)} ]$, and
$\mathit{rnd}_{n}^{*} \in \{ \mathit{\RNE}, \mathit{\RNZ},
\mathit{\RNN}, \mathit{\RNP} \}$. Similar to how we use
Lemma~\ref{lemma:rnd-sum-bound} to show that $\underline{P(x)}$
involves $\oplus_{\mathit{\RNN}}$ when $\mathit{flop}_{n} = \oplus$,
one can leverage $\otimes_{\mathit{\RNN}}$'s role in defining the
lower bound of FP multiplication detailed in
Lemma~\ref{lemma:rnd-prod-bound} to conclude that
$\mathit{rnd}_{n}^{*} = \mathit{\RNN}$ when $\mathit{flop}_{n} =
\otimes$. We thus focus on establishing that there cannot exist any
pair of operands $(p_{1}^{*}, p_{2}^{*})$ outside the Cartesian
product $\{ \underline{P_{1}(x)}, \overline{P_{1}(x)} \} \times \{
\underline{P_{2}(x)}, \overline{P_{2}(x)} \}$ such that $p_{1}^{*}
\otimes_{\mathit{\RNN}} p_{2}^{*} < min(\underline{P_{1}(x)}
\otimes_{\mathit{\RNN}} \underline{P_{2}(x)}, \underline{P_{1}(x)}
\otimes_{\mathit{\RNN}} \overline{P_{2}(x)}, \overline{P_{1}(x)}
\otimes_{\mathit{\RNN}} \underline{P_{2}(x)}, \overline{P_{1}(x)}
\otimes_{\mathit{\RNN}} \overline{P_{2}(x)})$. Henceforth, we refer to
the right hand side of $\underline{P(x)}$'s definition for the
$\mathit{flop}_{n} = \otimes$ case as
$\mathit{fp}\_\mathit{prod}_{\mathit{min}}$. One can deduce from the
rules of interval multiplication detailed in
Lemma~\ref{lemma:int-prod} that for any $p_{1} \in
[\underline{P_{1}(x)}, \overline{P_{1}(x)}]$ and $p_{2} \in
[\underline{P_{2}(x)}, \overline{P_{2}(x)}]$, the real arithmetic
product $p_{1} \times p_{2}$ is equal to or greater than
$\min(\underline{P_{1}(x)} \times \underline{P_{2}(x)},
\underline{P_{1}(x)} \times \overline{P_{2}(x)}, \overline{P_{1}(x)}
\times \underline{P_{2}(x)}, \overline{P_{1}(x)} \times
\overline{P_{2}(x)})$, which we subsequently refer to as
$\mathit{real}\_\mathit{prod}_{min}$. Let $p_{1,\mathit{real}}^{*}$
and $p_{2,\mathit{real}}^{*}$ be two FP numbers such that $(p_{1,
  \mathit{real}}^{*} , p_{2, \mathit{real}}^{*}) \in \{
\underline{P_{1}(x)}, \overline{P_{1}(x)}\} \times \{
\underline{P_{2}(x)}, \overline{P_{2}(x)} \}$ and
$p_{1,\mathit{real}}^{*} \times p_{2,\mathit{real}}^{*} =
\mathit{real}\_\mathit{prod}_{\mathit{min}}$. Given that $p_{1,
  \mathit{real}}^{*} \times p_{2, \mathit{real}}^{*} \leq p_{1} \times
p_{2}$ for any $p_{1} \in [ \underline{P_{1}(x)},
  \overline{P_{1}(x)}]$ and $p_{2} \in [ \underline{P_{2}(x)},
  \overline{P_{2}(x)}]$, the monotonic properties of $\otimes_{\RNN}$
detailed in Lemma~\ref{lemma:rnd-prod-monotonic} indicates that the
bound $p_{1, \mathit{real}}^{*} \otimes_{\mathit{\RNN}} p_{2,
  \mathit{real}}^{*} \leq p_{1} \otimes_{\mathit{\RNN}} p_{2}$ must
hold for all pairs of $p_{1}$ and $p_{2}$, including those outside the
aforementioned Cartesian product. This concludes our proof that no
pair of operands $(p_{1}^{*}, p_{2}^{*}) \notin \{
\underline{P_{1}(x)}, \overline{P_{1}(x)} \} \times \{
\underline{P_{2}(x)}, \overline{P_{2}(x)} \}$ can lead to an FP
product less than $\mathit{fp}\_\mathit{prod}_{min}$ when
$\mathit{rnd}_{n}^{*} = \mathit{\RNN}$ as well as our proof that
$\underline{P(x)} = \mathit{fp}\_\mathit{prod}_{\mathit{min}}$ when
$\mathit{flop}_{n} = \otimes$.

The proof for Theorem~\ref{theorem:ub-def} largely mirrors that of
Theorem~\ref{theorem:lb-def} due to the symmetry between them.  Akin
to the manner in which Lemmas~\ref{lemma:rnd-sum-bound}
and~\ref{lemma:rnd-prod-bound} were applied to justify
$\oplus_{\mathit{\RNN}}$'s and $\otimes_{\mathit{\RNN}}$'s roles in
defining $\underline{P(x)}$, the same lemmas can be applied to
validate the use of $\oplus_{\mathit{\RNP}}$ and
$\otimes_{\mathit{\RNP}}$ in defining $\overline{P(x)}$. As is the
case with Theorem~\ref{theorem:lb-def}, one can apply the interval
addition and multiplication rules detailed in
Lemmas~\ref{lemma:int-sum} and~\ref{lemma:int-prod} to justify the
operands associated with $\oplus_{\mathit{\RNP}}$ and
$\otimes_{\mathit{\RNP}}$ in Theorem~\ref{theorem:ub-def}'s definition
for the non-trivial cases.
\end{proof}

The polynomial generator from the \rlibm pipeline now has to solve the
stricter constraints resulting from Theorems~\ref{theorem:lb-def}
and~\ref{theorem:ub-def} to ensure correctness while incurring no
overheads in the final
implementations. Figure~\ref{fig:input_bounds_example} provides an
example of the steps taken by our polynomial generator to compute the
lower and upper bounds of each FP operation's result under various
rounding modes. It uses the \RNN mode to compute the lower bound and
then uses the \RNP mode to compute the upper bound based on the
specifications in Theorems~\ref{theorem:lb-def}
and~\ref{theorem:ub-def}.
%
%

\textbf{Accounting for rounding induced variability during reduced
  interval generation.}
\label{subsec:oc}
The polynomials produced by \rlibm's generators are approximations
over the variable $x'$ drawn from a reduced version of the target
function's original domain (\ie, its reduced range). Consequently,
each polynomial requires a set of operations that form the output
compensation function, which maps the polynomial outputs to the
original target range.  Within the \rlibm pipeline, the reduced
interval generator accounts for the effects of evaluating an output
compensation function via FP arithmetic, the result of which we denote
$OC(y',x)$ for any given $y'$. The goal of the reduced interval
generator is to find for each input $x$'s reduced input $x'$ the
maximal interval $[l',h']$ such that $\forall{y'}\in [l', h'], l \leq
\mathit{OC}(y',x) \leq h$. The values $l$ and $h$ denote the bounds of
$x$'s target rounding interval. By using the resultant $[l', h']$ as a
constraint for the polynomial generator, the \rlibm pipeline ensures
for a given reduced input $x'$ that a polynomial with an output
$P(x')$ that satisfies $l' \leq P(x') \leq h'$ also satisfies $l \leq
\mathit{OC}(P(x'),x) \leq h$. By satisfying the constraint $l \leq
\mathit{OC}(P(x'),x) \leq h$, $\mathit{OC}(P(x'),x)$ is guaranteed to correctly
round to the oracle result for the original input $x$.

The reduced interval generators in the original \rlibm pipeline
evaluate output compensation functions with the rounding mode set to
$\mathit{\RNE}$, thereby only accounting for the effects of rounding
under a specific mode. We note that the final post-output compensation
result $\mathit{OC}(y', x)$ for a specific instance of $y'$ can end up as a
different number within the range $[\underline{\mathit{OC}(y',x)},
  \overline{\mathit{OC}(y',x)}]$ depending on the rounding rule applied to each
component FP operation. As such, the goal of the interval generator
under the rounding-invariant input bounds approach is to find for each
$x'$ the maximal $[l', h']$ such that $\forall{y'} \in [l', h'], l
\leq \min_{y' \in [l', h']} \underline{\mathit{OC}(y',x)} \leq \max_{y' \in
  [l', h']} \overline{\mathit{OC}(y',x)} \leq h$. We note that every output
compensation function used by \rlibm can be expressed using only
addition and multiplication. In essence, the output compensation
functions with respect to the reduced inputs $x'$ are polynomial
compositions of the form $\hat{y}=Q(P(x'))$ evaluated with FP
additions and multiplications. The reduced interval generator can
therefore compute $min_{y' \in [l', h']} \underline{\mathit{OC}(y',x)}$ and
$\max_{y' \in [l', h']} \overline{\mathit{OC}(y',x)}$ by directly applying
Theorems \ref{theorem:lb-def} and \ref{theorem:ub-def}. By producing
constraints for the polynomial generator in this manner, the reduced
interval generator under the rounding-invariant input bounds approach
can validate the following statement: for all $P(x')$ such that $l'
\leq \underline{P(x')} \leq \overline{P(x')} \leq h'$, $l \leq
\min_{P(x') \in [\underline{P(x')},
    \overline{P(x')}]}\underline{OC(P(x'),x)} \leq \max_{P(x') \in
  [\underline{P(x')}, \overline{P(x')}]}\overline{OC(P(x'),x)} \leq h$
holds by default.
In summary, the reduced interval generator for the
rounding-invariant input bounds approach applies
Theorems~\ref{theorem:lb-def} and ~\ref{theorem:ub-def} to ensure that
an output compensation function can map polynomial outputs to their
final target rounding intervals under any rounding mode without
requiring rounding mode adjustments in the final implementations.

\textbf{Correctness under all faithful rounding modes.}
As previously established, $\underline{P(x)}$ is the smallest possible
value of $P(x)$ when any of the four rounding modes of concern can be
applied to any of the intermediate operations. $\overline{P(x)}$ is
the upper bound counterpart. The definitions of $\underline{P(x)}$ and
$\overline{P(x)}$ are founded on Lemma~\ref{lemma:rnd-bound}, which
establishes $\mathit{\RNN}$'s and $\mathit{\RNP}$'s roles as the lower
and upper bounds of the rounding modes considered. One must note that
their roles as bounds are not limited to the set $\{\mathit{\RNE},
\mathit{\RNZ}, \mathit{\RNN}, \mathit{\RNP}\}$ as
Lemma~\ref{lemma:rnd-bound} pertains to all faithful rounding
modes. Naturally, $\oplus_{\mathit{\RNN}}$ and $\oplus_{\mathit{\RNP}}$ (or
$\otimes_{\mathit{\RNN}}$ and $\otimes_{\mathit{\RNP}}$) each define the lower and upper
bounds of any faithful FP addition (or multiplication) operation given
the same pair of operands. We note that the range $[\underline{P(x)},
  \overline{P(x)}]$ subsumes every possible assignment of either
$\oplus_{\mathit{\RNN}}$, $\oplus_{\mathit{\RNP}}$, $\otimes_{\mathit{\RNN}}$, or
$\otimes_{\mathit{\RNP}}$ to each FP operation composing the evaluation of
$P(x)$.
As such, $\underline{P(x')}$ and $\overline{P(x')}$ are the bounds for
all the possible values of $P(x')$ when the $\mathit{rnd}_{i}$ applied
to each operation could be \textbf{ANY} faithful rounding mode
(i.e. round-away-from-zero, round-to-odd, etc). By combining the
polynomials generated under this new approach with range reduction and
output compensation algorithms that have rounding-invariant
correctness guarantees, we can create implementations that can produce
correct results under any faithful rounding mode~(without the overhead
of additional instructions required by our round-to-zero emulation
methods).
In essence, the new implementations built using the rounding-invariant
input bounds approach by default provide correctness guarantees for
applications using non-standard faithful rounding modes.

\section{Experimental Evaluation}
We report the results from experimentally evaluating the two proposed
methods with respect to correctness and performance relative to the
original \rlibm prototypes. 

\textbf{Prototype.}
We used the publicly available code from the \rlibm
project~\cite{rlibm-project} and built the two proposed methods. To
build the prototype that applies the rounding-invariant outputs
method's round-to-zero emulation, we changed the default rounding mode
for the \rlibm project's generators to the round-to-zero (\RNZ)
mode. To build the prototype for the rounding-invariant input bounds
approach, we made multiple changes: (1) changed the rounding-sensitive
range reduction operations to produce $\mathit{\RNZ}$ results, (2)
rewrote the reduced interval generation process to deduce new
intervals that account for rounding-induced variability in the output
compensation functions, and (3) updated the polynomial generator to
evaluate polynomials using interval bounds and generate polynomials
that satisfy stricter correctness constraints.
Our prototype uses the MPFR library~\cite{Fousse:toms:2007:mpfr} and
\rlibm's algorithm to compute the Oracle 34-bit round-to-odd (\RNO)
result for each input~\cite{lim:rlibmall:popl:2022}. Our prototype's
polynomial generator uses an exact rational arithmetic LP solver,
SoPlex, and \rlibm's publicly available randomized LP
solver~\cite{aanjaneya:rlibm-prog:pldi:2022,aanjaneya:hull-rlibm:pldi:2024}
to solve the constraints.
Using these methods, we have developed a new math library that has a
collection of twenty-four new implementations targeting twelve
elementary functions ($sin$, $sinh$, $sinpi$, $cos$, $cosh$, $cospi$,
$log$, $log2$, $log10$, $exp$, $exp2$, and $exp10$). They are designed
to directly produce correctly rounded results for all representations
with up to 32-bits with respect to all four standard rounding modes,
irrespective of the application's rounding mode. Our prototype is open
source, and the artifact is publicly
available~\cite{park:multiround-artifact:pldi:2025}.

\textbf{Methodology.} To evaluate correctness, we run the new
implementations on every input from each target representation
(10-bits to 32-bits). For every input, we run a given implementation
under all four rounding modes and compare the outputs rounded to the
target representation against the Oracle result generated using the
MPFR library. As a specific example, testing correctness for the
32-bit representation involves evaluating the results for each of the
4 billion inputs under all four rounding modes. To test the
prototype's ability to produce correct results while maintaining the
application-level rounding mode, we call \texttt{fesetround} with each
rounding mode before invoking the implementation being tested. To
compare the performance of the implementations built through the
different approaches, we use \texttt{rdtsc} to count the number of
cycles taken to compute the result for each 32-bit input. We aggregate
these counts to compute the total time taken by a given elementary
function implementation to produce outputs for the entire 32-bit FP
domain. We compile the test harnesses with the \texttt{-march=native
  -frounding-math -fsignaling-nans} flags. We performed the
experiments on Ubuntu 24.04 and an AMD EPYC 7313P CPU with a 3.0 GHz
base frequency.

%
\begin{wraptable}{r}{6.5cm}
\begin{small}  
    \begin{tabular}{|l|c|c|c|c|}
     \hline
      $f(x)$ & Ours & $\rlibm^{*}$ & Core-Math & glibc \\  
     \hline
     $ln(x)$ & \cmark & \xmark & \xmark & \xmark \\
     \hline
     $log_{2}(x)$ & \cmark & \xmark & \xmark & \xmark \\
     \hline
     $log_{10}(x)$ & \cmark & \xmark & \xmark & \xmark \\
     \hline
     $e^{x}$ & \cmark & \xmark & \xmark & \xmark \\
     \hline
     $2^{x}$ & \cmark & \xmark & \xmark & \xmark \\
     \hline
     $10^{x}$ & \cmark & \xmark & \xmark & \xmark \\
     \hline
     $sin(x)$ & \cmark & \xmark & \xmark & \xmark \\
     \hline
     $cos(x)$ & \cmark & \xmark & \xmark & \xmark \\
     \hline
     $sinh(x)$ & \cmark & \xmark & \xmark & \xmark \\
     \hline
     $cosh(x)$ & \cmark & \xmark & \xmark & \xmark \\
     \hline
     $sinpi(x)$ & \cmark & \xmark & \xmark & \NA \\
     \hline
     $cospi(x)$ & \cmark & \xmark & \xmark & \NA \\
     \hline
    \end{tabular}
    \caption{\small The table lists whether the libraries generate
      correctly rounded results for all inputs in FP representations
      with 10 to 32-bits under each of the four standard rounding
      modes. We compare the functions from our new library (both
      methods), \rlibm's library without rounding mode changes
      ($\rlibm^{*}$), Core-Math's 32-bit float libm, and glibc's
      64-bit double libm. We use \cmark \ to indicate that a function
      produces correct results for all inputs across all
      representations considered under all standard rounding modes. We
      use \xmark \ otherwise. \NA \ indicates that a function is not
      implemented in the tested library.}
    \label{tab:correctness_eval}
\end{small}
\end{wraptable}

\textbf{Ability to produce correctly rounded results with multiple
  rounding modes.} The 24 new functions built using our
rounding-invariant outputs and rounding-invariant input bounds
approaches produce correctly rounded results for all inputs in the
domain of every target representation with up to 32-bits across all
rounding modes $\mathit{rnd} \in
\{\mathit{\RNE},\mathit{\RNZ},\mathit{\RNN},\mathit{\RNP}\}$. Furthermore,
the new implementations directly produce correct results under the
application-level rounding mode without requiring calls to
\texttt{fesetround}. Table~\ref{tab:correctness_eval} reports the
ability of the libraries to produce correctly rounded results while
using the application's rounding mode.
For this evaluation, we disabled the rounding modes changes performed
by \rlibm's default library with \texttt{fesetround}, which is
indicated by $\rlibm^{*}$ in Table~\ref{tab:correctness_eval}. When we
do not change the rounding mode, $\rlibm^{*}$ produces incorrect
outputs for approximately 100 inputs. These incorrect outputs are
primarily due to rounding mode induced variability.
The \rlibm prototype with rounding mode changes to \RNE before each
call produces correctly rounded results for all inputs similar to our
methods but incurs performance overheads. CORE-MATH fails to produce
correctly rounded results for representations smaller than 32-bits due
to double rounding issues.  The math libraries for double-precision
from \texttt{glibc} do not produce correctly rounded results for any
representation. CORE-MATH and \texttt{glibc}'s double libm produce
incorrect (\ie, not a correctly rounded result) results for more than
200 million inputs for the \texttt{sin} function for a 31-bit
representation because of double rounding errors.

\textbf{Improved performance of generated math libraries.}
Figure~\ref{fig:performance_eval} reports the speedup (\ie, the ratio
of the execution times) of the implementations produced through our
rounding-invariant input bounds and rounding-invariant outputs
approaches over the default \rlibm functions (\ie, two bars for each
function). The default \rlibm implementations contain the latency of
the two separate calls to \texttt{fesetround}, one used to switch from
the application-level rounding mode to \RNE and one used to revert
back to the original rounding mode.
For all twelve functions, the implementation built using the
rounding-invariant input bounds approach exhibits the best performance
improvement.
On average, functions built using our rounding-invariant inputs method
are $2.3\times$ faster than the \rlibm functions.
Our rounding-invariant input bounds method produces implementations
that do not need any changes to the rounding mode while performing
similar amount of useful computation as the default \rlibm
implementations, which is the primary reason for performance
improvement. This significant performance improvement also highlights
the degree to which calls to \texttt{fesetround} degrade the
performance of the existing implementations.

In contrast to the rounding-invariant input bounds method, our
rounding-invariants outputs method that does round-to-zero emulation
is on average only $1.6\times$ faster than the original \rlibm
functions.
While the performance gains of the rounding-invariant
input bounds implementations are generally even across all the
functions considered, the results associated with the
rounding-invariant outputs method can vary significantly.
The cost of round-to-zero emulation depends on the number of additions
and multiplications in the final implementation, which depends on the
degree and the number of terms of the polynomial and can vary across
functions. The custom $\mathit{\RNZ}$ operations are alternatives to
the basic FP addition and multiplication operations and are primarily
used within the polynomial evaluation and output compensation portions
of the implementations. Each elementary function has a subset of
inputs for which the outputs can be directly approximated with a
constant, which obviates range reduction, polynomial evaluation, and
output compensation. The number of inputs that undergo polynomial
evaluation and output compensation varies widely between the different
functions. Consequently, the number of inputs subject to the overheads
introduced by the custom $\mathit{\RNZ}$ operations also varies
widely.

\begin{figure}
  \includegraphics[width=0.99\linewidth]{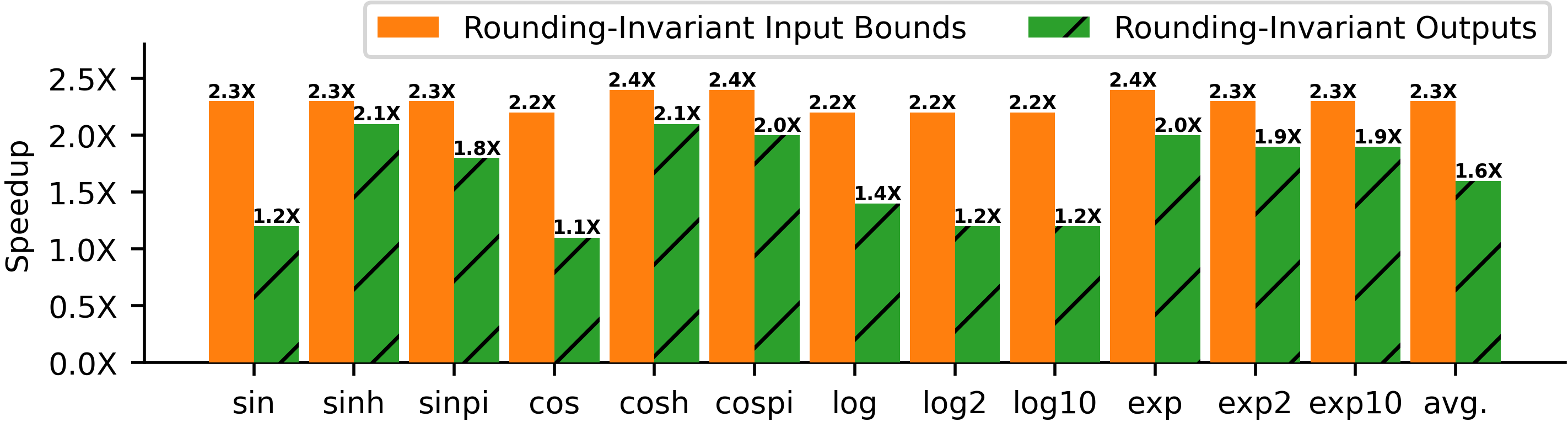}
  \caption{\small Speedup (\ie, ratio of the execution times) achieved
    by the implementations produced using our two approaches,
    rounding-invariant input bounds and rounding-invariant outputs,
    when compared to the original \rlibm counterparts.}
    \label{fig:performance_eval}    
\end{figure}

\textbf{Rounding mode changes with hardware instructions.}
The cost of changing the rounding mode depends on the architecture and
the implementation of the floating-point environment. While reporting
the performance improvements above, we use the \texttt{fesetround}
function because that is the most portable way to use the
implementations across architectures. On x86-64, the fesetround
function does three things: checks that the rounding mode is valid,
sets the rounding mode for the x87 environment, and sets the rounding
mode for the SSE instructions. Each call has an overall latency of 40
cycles on average.
If we can assume that the SSE rounding mode and the FP \texttt{x87}
rounding mode can be out of sync in the FP environment (\ie, any
interim fegetround call will be wrong) and the rounding mode is valid,
we can set the rounding mode by setting the \texttt{mxcsr} registers
in the \texttt{x86-64} ISA. Here is the snippet of the assembly code
that we used to change the rounding mode with just hardware
\texttt{x86-64} instructions.

\begin{verbatim}
unsigned int xcw;
__asm__("stmxcsr %0": "=m" (*&xcw);
xcw &= ~0x6000;
xcw |= round << 3;
__asm__("ldmxcsr %0": : "m" (*&xcw);
\end{verbatim}

The average total latency for changing the rounding mode with the
above snippet is 15 cycles. Our final implementations from the
rounding-invariant inputs bounds method are $1.4\times$ faster than
the original \rlibm's implementations using hardware assembly
instructions for changing the rounding mode.

In summary, the new functions generated using our two methods are not
only faster on average than their \rlibm counterparts but can also
produce correctly rounded results directly with the rounding modes
used by the applications.

\section{Related Work}

\textbf{Approximating elementary functions.} There is a long line of
work on approximating elementary functions for FP
representations~\cite{Abraham:fastcorrect:toms:1991,
  Daramy:crlibm:spie:2003,Fousse:toms:2007:mpfr,Muller:elemfunc:book:2016,Trefethen:chebyshev:book:2012,
  Remes:algorithm:1934,Dinechin:gappaverify:sac:2006,
  Dinechin:verify:tc:2011, Daumas:proofs:arith:2005,
  briggs:megalibm:popl:2024}. Range reduction is a key component of
this
task~\cite{Tang:log:toms:1990,Tang:TableLookup:SCA:1991,Tang:exp:toms:1989,762822,
  Cody:book:1980, Boldo:reduction:toc:2009}. Subsequently, the task is
to produce a polynomial approximation over the reduced domain.  A
well-known and popular tool is
Sollya~\cite{Chevillard:sollya:icms:2010}.
Using a modified Remez algorithm~\cite{Brisebarre:epl:arith:2007},
Sollya can generate polynomials with FP coefficients that minimize the
infinity norm. There is also subsequent work to compute and prove the
error bound on polynomial evaluation using interval
arithmetic~\cite{Chevillard:infnorm:qsic:2007,
  Chevillard:ub:tcs:2011}. Sollya is an effective tool for creating
polynomial approximations using the minimax method.
There have been efforts to prove bounds on the results of math
libraries~\cite{harrison:hollight:tphols:2009,
  Harrison:expproof:amst:1997, Harrison:verifywithHOL:tphol:1997,
  Sawada:verify:acl:2002, Lee:verify:popl:2018}. Recent efforts have
focused on repairing individual outputs of math
libraries~\cite{Xin:repairmlib:popl:2019,Daming:fpe:popl:2020}. Muller's
seminal book on elementary functions is an authoritative source on
this topic~\cite{Muller:elemfunc:book:2016}.

\textbf{Correctly rounded libraries.}
CR-LIBM~\cite{Daramy:crlibm:spie:2003} provides implementations of
functions that produce correctly rounded results for double-precision
for a single rounding mode.
CR-LIBM provides four distinct implementations, one
  corresponding to each rounding mode.
Further, when CR-LIBM's results are double rounded to target
representations, it can produce wrong results due to double rounding
errors~\cite{lim:rlibmall:popl:2022}. The CORE-MATH
project~\cite{sibidanov:core-math:arith:2022} is also building a
collection of correctly rounded elementary functions. It uses the
worst-case inputs needed for correct rounding and uses the error bound
required for those inputs while generating a minimax polynomial with
Sollya. However, they produce correctly rounded results for a specific
representation.

\textbf{Comparison to our prior work on \rlibm.}  This paper builds
upon our prior work in the \rlibm
project~\cite{Aanjaneya:2023:FPE,lim:rlibmall:popl:2022,lim:rlibm:popl:2021,lim:rlibm32:pldi:2021,aanjaneya:rlibm-prog:pldi:2022},
which approximates the correctly rounded result using an LP
formulation. We use \rlibm's method for generating oracles, \rlibm's
randomized LP algorithm for full rank systems, and its fast polynomial
generation. We use the \rlibm project's idea of approximating the
correctly rounded result for a 34-bit representation with the
round-to-odd mode to generate correctly rounded results for all inputs
with multiple representations and rounding modes with a single
polynomial approximation.  We enhance the \rlibm approach and the
entire pipeline to avoid the rounding mode changes necessary due to
the reliance on the round-to-nearest-mode as its implementation
rounding mode, which improves the performance by $2\times$.

An alternative to our approach is to develop four different
implementations using the \rlibm method where each implementation is
tailored to a specific rounding mode similar to CR-LIBM's
implementations.  It just requires checking the rounding mode and
choosing the correct implementation. In our interactions with various
math library developers, they were concerned about the code bloat and
concomitant software maintenance issues. Our collaborators at Intel,
who were interested in the \rlibm project, also favored a single
implementation that is usable as a reference library.
If the range reduction and output compensation methods are
rounding-mode oblivious, then it is possible to generate different
polynomial coefficients for each rounding mode rather than distinct
implementations. However, the range reduction and output compensation
functions used in the \rlibm project are not rounding-mode
oblivious~\cite{aanjaneya:rlibm-prog:arxiv:2021, Lim:rlibm:arxiv:2020,
  Lim:rlibm32:arxiv:2021, Lim:rlibmAll:arxiv:2021,
  park:trig:vss:2025}.  Hence, we use the rounding-invariant outputs
method to make range reduction rounding mode oblivious. Subsequently,
we use either the rounding-invariant outputs method or the
rounding-invariant input bounds method to address rounding-mode
induced variability in output compensation.
Further, our algorithms to produce the round-to-zero result from any
rounding mode can be independently useful in many applications.

\textbf{Accounting for error.}  We compute the error in an FP
operation to emulate the round-to-zero result using error-free
transformations, which have been used for compensated
summation~\cite{kahan:pracniques:cacm:1965,
  rump:AccurateSummation:2009}, compensated Horner
Scheme~\cite{Langlois:Horner:2005}, and robust geometric
algorithms~\cite{shewchuk:robust:1996}. The idea of Fast2Sum was used
in accurate summation by Kahan~\cite{kahan:pracniques:cacm:1965} and
Dekker~\cite{dekker:fast2sum:1971}. Fast2Sum is fast and requires
three FP operations and a branch instruction. Subsequently, it has
been modified to remove the branch with TwoSum~\cite{Knuth:1998:SA}.
Boldo~\etal~\cite{boldo:fast2sum:2017} have analyzed Fast2Sum and
TwoSum with respect to overflows and underflows. They have shown that
Fast2Sum is almost immune to overflow. The design of error-free
transformations for other rounding modes has been explored in a
dissertation by Priest~\cite{priest:thesis:1992}.
Dekker introduced an algorithm to identify the error in a
multiplication operation based on Veltkamp
splitting~\cite{muller:fp:2018}. Fused-multiply-add (FMA) instructions
make the computation of errors easier with just two
instructions~\cite{muller:fp:2018}.  This paper builds on these
results to design new decision procedures to produce the round-to-zero
result given an addition and a multiplication operation performed in
any rounding mode.

These error-free transformations (EFTs) have also been used recently
for debugging numerical code~\cite{demeure:shaman:thesis:2020,
  chowdhary:eftsanitizer:oopsla:2022}. Shaman~\cite{demeure:shaman:thesis:2020}
uses EFTs as an oracle and implements a C++ library using operator
overloading.  EFTSanitizer~\cite{chowdhary:eftsanitizer:oopsla:2022}
uses compile-time instrumentation to add these error-free
transformations for primitive operations and computes the rounding
error. It uses the MPFR library to measure the error in a call to an
elementary function. These tools are primarily focused on the
round-to-nearest rounding mode and did not explore other rounding
modes.

\section{Conclusion}
This paper proposes two methods, rounding-invariant outputs and
rounding-invariant input bounds, to design a single implementation
that produces correctly rounded results for all inputs with multiple
representations and rounding modes while using the application's
rounding mode. The key idea in the rounding-invariant outputs method
is to emulate the round-to-zero result for all rounding modes by
augmenting each FP addition and multiplication. We design new
algorithms to produce the round-to-zero result and provide their
associated correctness proofs. Through the rounding-invariant input
bounds method, we deduce the bounds on the output of a sequence of FP
operations to account for variability induced by rounding modes and
augment the \rlibm pipeline to incorporate these bounds.
Our math library can serve as a fast reference library for multiple
representations with up to 32-bits because it makes double rounding
innocuous. It is a step in our effort to make correct rounding
mandatory in the next version of the IEEE-754 standard and enhance the
portability of applications using such libraries.  In the future, we
want to explore extending this approach to develop correctly rounded
math libraries for the GPU ecosystem, the 64-bit representation, and
other extended precision representations.

\begin{acks}                            
  We thank the PLDI 2025 reviewers, Bill Zorn, and members of the
  Rutgers Architecture and Programming Languages (RAPL) lab for their
  feedback on this paper.  This material is based upon work supported
  in part by the research gifts from the Intel corporation and the
  \grantsponsor{GS100000001}{National Science
    Foundation}{http://dx.doi.org/10.13039/100000001} with grants:
  \grantnum{GS100000001}{2110861} and \grantnum{GS100000001}{2312220}.
  Any opinions, findings, and conclusions or recommendations expressed
  in this material are those of the authors and do not necessarily
  reflect the views of the Intel corporation or the National Science
  Foundation.
\end{acks}

\bibliography{references.bib}
\newpage
\section{Supplemental Material}
\label{sec:appendix}

In this section, we provide a detailed proof of correctness to
demonstrate that Algorithms~\ref{alg:rz_add} and ~\ref{alg:rz_mult}
produce the \RNZ results for addition and multiplication respectively,
which were omitted from our PLDI 2025
paper~\cite{park:rlibm-multiround:pldi:2025} due to space constraints.

Both algorithms rely on the sign of the rounding error in the original
operation (\ie, $a+b - (a \oplus_{\mathit{rnd}} b)$ or $a \times b -
(a \otimes_{\mathit{rnd}} b)$) to determine whether the initial sum or
product obtained through the application's rounding mode needs to be
adjusted to conform to $\RNZ$. Because the rounding errors caused by
$\oplus_{\mathit{rnd}}$ and $\otimes_{\mathit{rnd}}$ are not
guaranteed to be exactly representable as FP numbers, both $\RZA$ and
$\RZM$ leverage \textit{floating-point approximations} of the rounding
errors induced by their respective target operations. In order to
correctly apply theorems that utilize the signs of the $\textit{real}$
rounding errors (\ie Theorems ~\ref{theorem:rz-add-1-proof} and
~\ref{theorem:rz-mult-1-proof}), we must ensure that the FP rounding error
approximations computed by $\RZA$ and $\RZM$ have the same signs as
their real number counterparts.  To prove that the rounding error
approximations used in our algorithms preserve the signs of their real
number counterparts, we rely on
Lemma~\ref{lemma:sign-preservation-proof}, which we present again for
reference.

\setcounter{theorem}{2}

\begin{definition}\label{rnd-sign-proof}  
For all $v \in (\mathbb{R} \setminus \{0\}) \cup \mathbb{T}$ where $\mathbb{T}$ is a
  FP representation, we define $\mathit{sign}(v)$ to be $0$ for positive
  numbers and $1$ for negative numbers. For the FP numbers $+0,-0 \in \mathbb{T}$, we
  define the sign as $\mathit{sign}(+0) = 0$ and $\mathit{sign}(-0) = 1$.
\end{definition}

\setcounter{lemma}{3}

\begin{lemma}\label{lemma:sign-preservation-proof}
Let $rnd$ be any rounding function that faithfully
  rounds a number $r \in \mathbb{R} \setminus \{ 0 \}$ to a FP number
  $t \in \mathbb{T}$. For all $r \in \mathbb{R}\setminus\{0\}$ and for
  all $rnd$, $\mathit{sign}(r) = \mathit{sign}(\mathit{rnd}(r))$.
\end{lemma}

With respect to our theorems for $\RZA$ and $\RZM$, the real values $r \in \mathbb{R} \setminus \{ 0 \}$ of concern are either the real arithmetic product or sum of two non-$NaN$, non-infinity FP operands.
Similarly, the corresponding FP numbers $rnd(r)$ of concern are the outputs of FP multiplication or addition operations between two non-$NaN$, non-infinity FP operands (\ie, $a \otimes_{rnd} b$ or $a \oplus_{rnd} b$).
We therefore prove Lemma ~\ref{lemma:sign-preservation-proof} in the context of our paper by showing how the operations $\otimes_{rnd}$ and $\oplus_{rnd}$ preserve signs for the faithful 
rounding modes $rnd \in \{ \RNE, \RNZ, \RNN, \RNP \}$. We do so using two additional lemmas (see Lemmas ~\ref{lemma:sign-preservation-mult} and ~\ref{lemma:sign-preservation-add} below), which each cover multiplication and addition.

\setcounter{lemma}{10}

\begin{lemma}\label{lemma:sign-preservation-mult}
Let $a, b \in \mathbb{T} \setminus \{\mathit{NaN}, \pm \infty \}$ be two FP numbers such that $a \times b \neq 0$.  For all faithful rounding modes $rnd \in \{ \RNE, \RNZ, \RNN, \RNP \}$, $\mathit{sign}(a \times b) = \mathit{sign}(a \otimes_{\mathit{rnd}} b)$.
\end{lemma}
\begin{proof}
For all $rnd \in \{ \RNE, \RNZ, \RNN, \RNP \}$, we require the FP multiplication operation $\otimes_{\mathit{rnd}}$ as defined in Equation~7 in Figure~\ref{fig:background:ops}(b) to adhere to the IEEE-754 standard.
The IEEE-754 2019 standard states that the sign bit of a FP multiplication result is the \textit{exclusive or} of the sign bits of its operands under all rounding modes.
Given our definition of signs in Definition ~\ref{rnd-sign-proof}, the sign of any non-zero, real arithmetic product of two FP numbers is also the exclusive or of the signs of its operands.
This is because a real arithmetic product can be only be negative when its operands have different signs. 
Under Definition ~\ref{rnd-sign-proof}, adherence to the IEEE-754 standard guarantees that  $\mathit{sign}(a \otimes_{\mathit{rnd}} b) = \mathit{sign}(a \times b)$ 
for any $rnd \in \{ \RNE, \RNZ, \RNN, \RNP \}$ and any non-$NaN$, non-infinity FP numbers $a$ and $b$ such that $a \times b \neq 0$.

\end{proof}

\begin{lemma}\label{lemma:sign-preservation-add}
Let $a, b \in \mathbb{T} \setminus \{\mathit{NaN}, \pm \infty \}$ be two FP numbers such that $a+b \neq 0$.  For all faithful rounding modes $rnd \in \{ \RNE, \RNZ, \RNN, \RNP \}$, $\mathit{sign}(a + b) = \mathit{sign}(a \oplus_{\mathit{rnd}} b)$.
\end{lemma}
\begin{proof}
We reiterate that we require the FP addition operation $\oplus_{rnd}$ as defined in Equation ~6 in Figure~\ref{fig:background:ops}(a) to adhere to the IEEE-754 standard and perform faithful rounding.
Given such chracteristics of $\oplus_{rnd}$, we prove Lemma ~\ref{lemma:sign-preservation-add} via contradiction. Doing so requires an additional lemma detailing the expected bounds on the sum of two FP numbers, 
which we present below along with its proof.
\end{proof}

\begin{lemma}\label{lemma:rz-add-bound}
Let $a$ and $b$ be two non-NaN, non-infinity FP numbers. Let $e_{min}$ and $p$ denote the minimum exponent and available precision of the target representation respectively. If $a+b \neq 0$, then $ 2^{e_{min}-p+1} \leq |a+b|$.
\end{lemma}

\begin{proof}
By their definitions, $a$ and $b$ are both FP numbers that are integer multiples of $2^{e_{min}-p+1}$, the smallest positive number in the representation. Let $a=c_{1} \times 2^{e_{min}-p+1}$ and $b=c_{2} \times 2^{e_{min}-p+1}$ where $c_{1}$ and $c_{2}$ are both integers. Based on these definitions, $a+b = (c_{1}+c_{2}) \times 2^{e_{min}-p+1}$ where $c_{1}+c_{2}$ is also an integer. If $a+b \neq 0$, then $c_{1}+c_{2} \neq 0$ and thus $1 \leq |c_{1}+c_{2}|$. The lower bound for $|c_{1}+c_{2}|$ indicates that when $a+b \neq 0$, $|a+b|=|(c_{1}+c_{2}) \times 2^{e_{min}-p+1}| \geq 2^{e_{min}-p+1}$.
\end{proof}

Given two FP numbers $a$ and $b$ such that $a+b \neq 0$, suppose that $\mathit{sign}(a+b) = 0$ and $\mathit{sign}(a \oplus_{\mathit{rnd}} b) = 1$  for some rounding mode $rnd \in \{ \RNE, \RNZ, \RNN, \RNP \}$. 
Under our definition of $\mathit{sign}$ in Definition ~\ref{rnd-sign-proof}, these conditions imply $a+b$ is positive while $a \oplus_{\mathit{rnd}} b$ is either $-0$ or a negative number. Applying Lemma ~\ref{lemma:rz-add-bound}, one can infer that $a+b \ge 2^{e_{min}-p+1}$ since $a+b > 0$. 
The assumption $\mathit{sign}(a \oplus_{\mathit{rnd}} b) = 1$ would thus contradict $a \oplus_{rnd} b$ being a faithful rounding of $a+b$ because the FP number $2^{e_{min}-p+1}$ is closer to all positive real numbers than $-0$ or any negative FP number.
Similarly, the case in which $\mathit{sign}(a+b) = 1$ and $\mathit{sign}(a \oplus_{\mathit{rnd}} b) = 0$ also leads to a contradiction because the FP number $-2^{e_{min}-p+1}$ is closer to all negative real numbers than $+0$ or any positive FP number.
The contradictions derivable from Lemma ~\ref{lemma:rz-add-bound} thus prove that $\mathit{sign}(a \oplus_{\mathit{rnd}} b) = \mathit{sign}(a + b)$ 
for any $rnd \in \{ \RNE, \RNZ, \RNN, \RNP \}$ and any non-$NaN$, non-infinity FP numbers $a$ and $b$ such that $a + b \neq 0$

By proving Lemmas ~\ref{lemma:sign-preservation-mult} and ~\ref{lemma:sign-preservation-add}, we have proven Lemma ~\ref{lemma:sign-preservation-proof} for the non-zero real numbers relevant to our 
algorithms for $RZA$ and $RZM$ - real arithmetic products or sums of two non-$NaN$, non-infinity FP numbers. Having established a foundational lemma for the subsequent theorems, we move on to our proofs for $\RZA$ and $\RZM$.

\subsection{Proof of Correctness for $\RZA$}

\label{sec:round-to-zero-add}
 
After handling the corner case where $a \oplus_{\RNZ} b \neq a
\oplus_{\RNN} b$ when $a = -b$ (lines 2 to 4), the primary task of
$\RZA$ in Algorithm~\ref{alg:rz_add} is to check the condition $a+b
\neq a \oplus_{\mathit{rnd}} b$ without having direct access to the real value
$a+b$. At a high-level, the algorithm computes a proxy for the error
term $a+b - (a \oplus_{\mathit{rnd}} b)$ and checks for $a+b - (a \oplus_{\mathit{rnd}}
b) \neq 0$ to confirm $a+b \neq a \oplus_{\mathit{rnd}} b$. The condition $a+b
- (a \oplus_{\mathit{rnd}} b) \neq 0$ indicates not only the presence of
rounding error in $a \oplus_{\mathit{rnd}} b$, but also that $a+b \neq 0$ since
$a \oplus_{\mathit{rnd}} b$ would be either $+0$ or $-0$ and $a+b - (a
\oplus_{\mathit{rnd}} b)$ would subsequently be equal to 0 in such a case. Once
it is established that $a+b - (a \oplus_{\mathit{rnd}} b) \neq 0$, the sign of
$a+b - (a \oplus_{\mathit{rnd}} b)$ relative to $a \oplus_{\mathit{rnd}} b$
concomitantly serves as an indicator for the condition $|a+b| < |a
\oplus_{\mathit{rnd}} b|$ as we show in our proof of Theorem
~\ref{theorem:rz-add-1-proof}, which we state below.

\setcounter{theorem}{3}

\begin{theorem}\label{theorem:rz-add-1-proof}
Let $a$ and $b$ be two non-NaN, non-infinity floating-point numbers such that $a \oplus_{\mathit{rnd}} b$ does not overflow for any rounding mode. If $a+b - (a \oplus_{\mathit{rnd}} b) \neq 0$, $a \oplus_{\mathit{rnd}} b$ and $a+b - (a \oplus_{\mathit{rnd}} b)$ have different signs if and only if $|a + b| < |a \oplus_{\mathit{rnd}} b|$.
\end{theorem}

\begin{proof}
We prove via contradiction that when $a+b - (a \oplus_{\mathit{rnd}} b) \neq 0$, $\mathit{sign}(a \oplus_{\mathit{rnd}} b) \neq \mathit{sign}(a+b - (a \oplus_{\mathit{rnd}} b))$ implies $|a+b| < |a \oplus_{\mathit{rnd}} b|$. Suppose that $a+b - (a \oplus_{\mathit{rnd}} b) \neq 0$, $\mathit{sign}(a \oplus_{\mathit{rnd}} b) \neq \mathit{sign}(a+b - (a \oplus_{\mathit{rnd}} b))$, and $|a \oplus_{\mathit{rnd}} b| < |a+b|$. With these assumptions, one can conclude that $a+b - (a \oplus_{\mathit{rnd}} b)$ will have the same sign as $a+b$ (\ie, $\mathit{sign}(a+b - (a \oplus_{\mathit{rnd}} b)) = \mathit{sign}(a+b)$) given that $|a \oplus_{\mathit{rnd}} b| < |a+b|$. As $a+b - (a \oplus_{\mathit{rnd}} b) \neq 0$, $a+b$ must be a non-zero value. Given the definition provided in Equation 6 (Figure ~\ref{fig:background:ops}) for such a case, $a \oplus_{\mathit{rnd}} b = \mathit{rnd}(a + b)$. Since all faithful rounding functions preserve the signs of non-zero values as detailed in Lemma ~\ref{lemma:sign-preservation-proof}, $\mathit{sign}(a \oplus_{\mathit{rnd}} b) = \mathit{sign}(a + b)$. Since $\mathit{sign}(a+b - (a \oplus_{\mathit{rnd}} b)) = \mathit{sign}(a+b)$, $\mathit{sign}(a \oplus_{\mathit{rnd}} b) = \mathit{sign}(a+b - (a \oplus_{\mathit{rnd}} b))$ holds by transitive equality, thereby directly contradicting the assumption that $\mathit{sign}(a \oplus_{\mathit{rnd}} b) \neq \mathit{sign}(a + b - (a \oplus_{\mathit{rnd}} b))$. The proposition in the reverse direction (\ie, assuming $a+b - (a \oplus_{\mathit{rnd}} b) \neq 0$, if $|a+b| < |a \oplus_{\mathit{rnd}} b|$ then $\mathit{sign}(a \oplus_{\mathit{rnd}} b) \neq \mathit{sign}(a+b - (a \oplus_{\mathit{rnd}} b))$) can also be proven via contradiction.  Suppose that $a+b - (a \oplus_{\mathit{rnd}} b) \neq 0$, $|a+b| < |a \oplus_{\mathit{rnd}} b|$, and $\mathit{sign}(a \oplus_{\mathit{rnd}} b) = \mathit{sign}(a+b - (a \oplus_{\mathit{rnd}} b))$. If $|a+b| < |a \oplus_{\mathit{rnd}} b|$, then $a+b -(a \oplus_{\mathit{rnd}} b)$ must have the sign of $-(a \oplus_{\mathit{rnd}} b)$. The implied sign of $a+b - (a \oplus_{\mathit{rnd}} b)$ directly contradicts the assumption that $\mathit{sign}(a \oplus_{\mathit{rnd}} b) = \mathit{sign}(a+b - (a \oplus_{\mathit{rnd}} b))$, thereby concluding our proof.
\end{proof}
Based on Theorem \ref{theorem:rz-add-1-proof}, one can conclude that testing whether $a+b-(a \oplus_{\mathit{rnd}} b) \neq 0$ and whether $\mathit{sign}(a+b-(a \oplus_{\mathit{rnd}} b)) \neq \mathit{sign}(a \oplus_{\mathit{rnd}} b)$ is sufficient for determining if the original sum (\ie, $a \oplus_{\mathit{rnd}} b$) needs to be augmented in accordance with Equation~\ref{get-rz-add}. Having established what conditions to check for, we now prove $\RZA$'s correctness by confirming that it correctly applies Theorem ~\ref{theorem:rz-add-1-proof} to ascertain whether $|a+b| < |a \oplus_{\mathit{rnd}} b|$. Lines 5 through 10 in Algorithm \ref{alg:rz_add} employ the steps in Dekker's $FastTwoSum$ algorithm ~\cite{dekker:fast2sum:1971} to compute the value $t$, which is a FP approximation of $a+b-(a \oplus_{\mathit{rnd}} b)$. Because all operations are executed in FP arithmetic, $s = a \oplus_{\mathit{rnd}} b$, $z=s \oplus_{\mathit{rnd}} (-a) = (a \oplus_{\mathit{rnd}} b) \oplus_{\mathit{rnd}} (-a) $ and $t= b \oplus_{\mathit{rnd}} (-z) = b \oplus_{\mathit{rnd}} (- ((a \oplus_{\mathit{rnd}} b) \oplus_{\mathit{rnd}} (-a)))$. After computing $t$, Algorithm \ref{alg:rz_add} applies Theorem ~\ref{theorem:rz-add-1-proof} by assessing whether $a+b-(a \oplus_{\mathit{rnd}} b) \neq 0$ and $\mathit{sign}(a+b- (a\oplus_{\mathit{rnd}}b)) \neq \mathit{sign}(a \oplus_{\mathit{rnd}} b)$ through the comparisons $bit(t) << 1 \neq 0$ and $(bit(t)$ \textbf{xor} $bit(s)) \geq 0x8000000000000000$ respectively. The comparison $bit(t) << 1 \neq 0$ checks if $t$ is neither $+0$ nor $-0$ while $(bit(t)$ \textbf{xor} $bit(s)) \geq 0x8000000000000000$ checks if $\mathit{sign}(t) \neq \mathit{sign}(a \oplus_{\mathit{rnd}} b)$. In essence, Algorithm ~\ref{alg:rz_add} uses $t$ as a proxy for $a+b - (a \oplus_{\mathit{rnd}} b)$ for the purposes of applying Theorem ~\ref{theorem:rz-add-1-proof}. Therefore, proving the correctness of Algorithm ~\ref{alg:rz_add} is contingent on establishing the following properties for $t$: (1) $t$ is neither $+0$ nor $-0$ if and only if $a+b - (a \oplus_{\mathit{rnd}} b) \neq 0$, and (2) $\mathit{sign}(t) = \mathit{sign}(a+b - (a \oplus_{\mathit{rnd}} b))$ whenever $a + b - (a \oplus_{\mathit{rnd}} b) \neq 0$. Establishing such properties for $t$ requires affirming that $t$ is a faithful rounding of $a+b - (a \oplus_{\mathit{rnd}} b)$.

Analysis of the $FastTwoSum$ algorithm by Boldo~\etal~\cite{boldo:fast2sum:2017} shows that when the exponent of the non-$NaN$ FP number $a$ (\ie, $e_{a}$) is equal to or greater than that of the non-$NaN$ FP number $b$ (\ie, $e_{b}$), $(a \oplus_{\mathit{rnd}} b)-a$ is exactly representable under any rounding mode, assuming the absence of overflow in $a \oplus_{\mathit{rnd}} b$. The intended inputs for $\RZA$ are non-$NaN$, non-infinity FP numbers and satisfy the last condition. Lines 6 through 8 ensure that $e_{a}\geq e_{b}$ by the time $z = (a \oplus_{\mathit{rnd}} b) \oplus_{\mathit{rnd}} (-a)$ is computed. Therefore, our implementation of $\RZA$ satisfies at runtime all the conditions necessary for $(a \oplus_{\mathit{rnd}} b)-a$ to be exactly representable as a FP number. With $(a \oplus_{\mathit{rnd}} b)-a$ being exactly representable, $z= (a \oplus_{\mathit{rnd}} b) \oplus_{\mathit{rnd}} (-a) = (a \oplus_{\mathit{rnd}} b) - a$, and thus the equality $t= b \oplus_{\mathit{rnd}} (-z) = b \oplus_{\mathit{rnd}} (a - (a \oplus_{\mathit{rnd}} b))$ will hold under all rounding modes. Given that the FP operation $\oplus_{\mathit{rnd}}$ adheres to faithful rounding, we conclude from Boldo ~\etal's analysis that $t = b \oplus_{\mathit{rnd}} (a - (a \oplus_{\mathit{rnd}} b))$ is a faithfully rounded version of $b + a - (a \oplus_{\mathit{rnd}} b) = a + b - (a \oplus_{\mathit{rnd}} b)$ under all faithful rounding modes $\mathit{rnd}$.

Boldo~\etal~\cite{boldo:fast2sum:2017} prove that $t$ is a faithful rounding of the FP addition error  $a+b-(a \oplus_{\mathit{rnd}} b)$ under the previously mentioned conditions. However, they conclude that $a+b-(a \oplus_{\mathit{rnd}} b)$ is not guaranteed to be  exactly representable as a FP number when $\mathit{rnd} \neq \RNE$. In other words, the equality $t = a+b-(a \oplus_{\mathit{rnd}} b)$ is not guaranteed for all rounding modes. Consequently, the desired properties of $t$ (\ie, $t$ is neither $+0$ nor $-0$ if and only if $a+b - (a \oplus_{\mathit{rnd}} b) \neq 0$ and $\mathit{sign}(t) = \mathit{sign}(a+b-(a \oplus_{\mathit{rnd}} b))$ whenever $a + b - (a \oplus_{\mathit{rnd}} b) \neq 0$) cannot be immediately assumed. Nevertheless, we can show that $t$'s status as a faithful rounding of $a+b - (a \oplus_{\mathit{rnd}} b)$ is sufficient for our purposes. Specifically, we leverage the definition of faithful rounding to prove Theorem ~\ref{theorem:rz-add-2-proof}, which we present below.

\begin{theorem}\label{theorem:rz-add-2-proof}
Let $t \in \mathbb{T}$ be a faithful rounding of the floating-point addition error $a+b - (a \oplus_{\mathit{rnd}} b) \in \mathbb{R}$. $t$ is neither $+0$ nor $-0$ if and only if $a+b - (a \oplus_{\mathit{rnd}} b) \neq 0$.  
\end{theorem}

\begin{proof}
The proof that $t$ is neither $+0$ nor $-0$ implies $a+b-(a \oplus_{\mathit{rnd}} b)\neq0$ only requires the definition of faithful rounding.
Given the relationship between the two numbers, suppose that $t$ is neither $+0$ nor $-0$ and $a+b-(a \oplus_{\mathit{rnd}} b) = 0$. These two propositions cannot simultaneously be true under any faithful rounding mode because the number $0$ can be exactly represented as either $+0$ or $-0$. The proposition that $a+b-(a \oplus_{\mathit{rnd}} b) \neq0$ implies $t$ is neither $+0$ nor $-0$ is also provable via contradiction, but requires the following lemma.
\end{proof}

\begin{lemma}\label{lemma:rz-add-error-bound}
Let $a$ and $b$ be two non-NaN, non-infinity FP numbers such that $a \oplus_{\mathit{rnd}} b$ does not overflow for any rounding mode. Let $e_{min}$ and $p$ denote the minimum exponent and available precision of the target representation respectively. If $a+b-(a \oplus_{\mathit{rnd}} b) \neq 0$, then $ 2^{e_{min}-p+1} \leq |a+b-(a \oplus_{\mathit{rnd}} b)|$.
\end{lemma}

\begin{proof}
The proof for this lemma is nearly identical to that for Lemma ~\ref{lemma:rz-add-bound}. By their definitions, $a$, $b$, $a \oplus_{\mathit{rnd}} b$ are all FP numbers that are integer multiples of $2^{e_{min}-p+1}$, the smallest positive number in the representation. Let $a=c_{1} \times 2^{e_{min}-p+1}$, $b=c_{2} \times 2^{e_{min}-p+1}$, and  $a \oplus_{\mathit{rnd}} b =c_{3} \times 2^{e_{min}-p+1}$ where $c_{1}$, $c_{2}$, and $c_{3}$ are all integers. Based on these definitions, $a+b - (a \oplus_{\mathit{rnd}} b) = (c_{1}+c_{2} - c_{3}) \times 2^{e_{min}-p+1}$ where $c_{1}+c_{2}-c_{3}$ is also an integer. If $a+b - (a \oplus_{\mathit{rnd}} b) \neq 0$, then $c_{1}+c_{2} - c_{3} \neq 0$ and thus $1 \leq |c_{1}+c_{2}-c_{3}|$. The lower bound for $|c_{1}+c_{2}-c_{3}|$ indicates that when $a+b - (a \oplus_{\mathit{rnd}} b) \neq 0$, $|a+b-(a \oplus_{\mathit{rnd}} b)|=|(c_{1}+c_{2}-c_{3}) \times 2^{e_{min}-p+1}| \geq 2^{e_{min}-p+1}$.
\end{proof}

Given Lemma ~\ref{lemma:rz-add-error-bound}, suppose $a + b - (a
\oplus_{\mathit{rnd}} b) \neq 0$ and $t$ is either $+0$ or
$-0$. Lemma~\ref{lemma:rz-add-error-bound} indicates that when
$0<a+b-(a \oplus_{\mathit{rnd}} b)$, the FP number $2^{e_{min}-p+1}$
is closer to $a+b-(a \oplus_{\mathit{rnd}} b)$ than both $+0$ and
$-0$. Similarly, when $a+b-(a \oplus_{\mathit{rnd}} b)<0$,
$-2^{e_{min}-p+1}$ would be a closer FP number than both $+0$ and
$-0$. Having previously established that $t$ is a faithful rounding of
$a+b - (a \oplus_{\mathit{rnd}} b)$, it must be the case that
$2^{e_{min}-p+1} \leq |t|$ when $a+b -(a \oplus_{\mathit{rnd}} b) \neq
0$. The resulting lower bound on $t$'s magnitude directly contradicts
the assumption that $t$ is either $+0$ or $-0$. The contradiction
derivable from Lemma ~\ref{lemma:rz-add-error-bound} concludes the
proof for Theorem \ref{theorem:rz-add-2-proof}.

Theorem \ref{theorem:rz-add-2-proof} ensures that $\RZA$ can correctly test whether $a+b - (a \oplus_{\mathit{rnd}} b) \neq 0$ through the comparison $bit(t)<<1 \neq 0$, which checks if $t$ is neither $+0$ nor $-0$. The last step in proving the correctness of $\RZA$ is establishing that the comparison ($bit(t) \ \textbf{xor} \ bit(s)) \geq 0x8000000000000000$, which checks if $\mathit{sign}(t) \neq \mathit{sign}(s = a \oplus_{\mathit{rnd}} b)$, is sufficient for determining that $\mathit{sign}(a+b - (a \oplus_{\mathit{rnd}} b)) \neq \mathit{sign}(a \oplus_{\mathit{rnd}} b)$. We corroborate $\RZA$'s usage of ($bit(t) \ \textbf{xor} \ bit(s)) \geq 0x8000000000000000$ by showing that $\mathit{sign}(t) = \mathit{sign}(a+b - (a \oplus_{\mathit{rnd}} b))$ whenever $a + b - (a \oplus_{\mathit{rnd}} b) \neq 0$. As previously mentioned, $t$ is a faithful rounding of $a + b - (a \oplus_{\mathit{rnd}} b)$. When $a + b - (a \oplus_{\mathit{rnd}} b) \neq 0$, $t$ would be a faithfully rounded version of a non-zero real number, and thus $t = \mathit{rnd}(a + b - (a \oplus_{\mathit{rnd}} b))$ for any rounding function $\mathit{rnd}$ that performs faithful rounding. Assuming $a + b - (a \oplus_{\mathit{rnd}} b) \neq 0$, one can apply Lemma ~\ref{lemma:sign-preservation-proof} to conclude that $\mathit{sign}(t) = \mathit{sign}(a+b - (a \oplus_{\mathit{rnd}} b))$.

Having validated the two key properties of $t$, we now summarize how
$\RZA$ correctly applies Theorem~\ref{theorem:rz-add-1-proof} through
the comparisons $bit(t) << 1 \neq 0$ and $(bit(t)$ \textbf{xor}
$bit(s)) \geq 0x8000000000000000$ in line 11 of Algorithm
~\ref{alg:rz_add}. Because $t$ is neither $+0$ nor $-0$ if and only if
$a+b - (a \oplus_{\mathit{rnd}} b) \neq 0$, the comparison $bit(t) <<
1 \neq 0$, which checks if $t$ is neither $+0$ nor $-0$, can confirm
whether $a+b - (a \oplus_{\mathit{rnd}} b) \neq 0$. The algorithm can
thus guarantee that $a+b - (a \oplus_{\mathit{rnd}} b) \neq 0$ if the
condition $bit(t) << 1 \neq 0$ is true. Because $\mathit{sign}(t) =
\mathit{sign}(a+b-(a \oplus_{\mathit{rnd}} b))$ when $a+b - (a
\oplus_{\mathit{rnd}} b) \neq 0$, we are able to assert the following:
The comparison $(bit(t)$ \textbf{xor} $bit(s)) \geq
0x8000000000000000$, which inspects if $\mathit{sign}(t) \neq
\mathit{sign}(s = a \oplus_{\mathit{rnd}} b)$, can confirm that
$\mathit{sign}(a+b-(a \oplus_{\mathit{rnd}} b)) \neq \mathit{sign}(a
\oplus_{\mathit{rnd}} b)$ if $bit(t) << 1 \neq 0$. Therefore, we
conclude that $a+b - (a \oplus_{\mathit{rnd}} b) \neq 0$ and
$\mathit{sign}(a+b - (a \oplus_{\mathit{rnd}} b)) \neq \mathit{sign}(a
\oplus_{\mathit{rnd}} b)$ when both $bit(t) << 1 \neq 0$ and $(bit(t)$
\textbf{xor} $bit(s)) \geq 0x8000000000000000$ are true. $\RZA$ can
thus accurately apply Theorem ~\ref{theorem:rz-add-1-proof} through
the comparisons in line 11 of Algorithm ~\ref{alg:rz_add} to determine
that $|a+b| < |a \oplus_{\mathit{rnd}} b|$ and subsequently adjust $a
\oplus_{\mathit{rnd}} b$ to match the $\RNZ$ result through line 12.
In conclusion, the guarantees of the $FastTwoSum$ algorithm and the
theorems proven above ensure that $\RZA$ produces $a \oplus_{RZ} b$
across all rounding modes for all non-$NaN$, non-infinity FP inputs
$a$ and $b$ such that $a \oplus_{\mathit{rnd}} b$ does not overflow.

\subsection{Proof of Correctness for $\RZM$}\label{sec:round-to-zero-mult}
We now prove that the function $\RZM$ detailed in Algorithm~\ref{alg:rz_mult} produces the \RNZ result for multiplication under all rounding modes. The core strategy in Algorithm~\ref{alg:rz_mult} is to confirm $a \times b - (a \otimes_{\mathit{rnd}} b) \neq 0$ as a means of verifying that $a \times b$ is not exactly representable. Similarly to $\RZA$, $\RZM$ first determines if the FP multiplication error $a \times b - (a \otimes_{\mathit{rnd}} b)$ is non-zero and then examines its sign relative to the sign of the FP product $a \otimes_{\mathit{rnd}} b$ to determine whether $|a \times b| < |a \otimes_{\mathit{rnd}} b|$. $\RZM$'s utilization of $a \times b - (a \otimes_{\mathit{rnd}} b)$ is founded on Theorem ~\ref{theorem:rz-mult-1-proof}, which we prove below.

\begin{theorem}\label{theorem:rz-mult-1-proof}
Let $a$ and $b$ be two non-$NaN$, non-infinity floating-point numbers such that $a \otimes_{\mathit{rnd}} b$ does not overflow for any rounding mode. If $a \times b - (a \otimes_{\mathit{rnd}} b) \neq 0$, $a \otimes_{\mathit{rnd}} b$ and $a \times b - (a \otimes_{\mathit{rnd}} b)$ have different signs if and only if $|a \times b| < |a \otimes_{\mathit{rnd}} b|$.
\end{theorem}

\begin{proof}
 We first prove that when $a \times b - (a \otimes_{\mathit{rnd}} b) \neq 0$, $\mathit{sign}(a \otimes_{\mathit{rnd}} b) \neq \mathit{sign}(a \times b - (a \otimes_{\mathit{rnd}} b))$ implies $|a \times b| < |a \otimes_{\mathit{rnd}} b|$. Suppose $a \times b - (a \otimes_{\mathit{rnd}} b) \neq 0$, $\mathit{sign}(a \otimes_{\mathit{rnd}} b) \neq \mathit{sign}(a \times b - (a \otimes_{\mathit{rnd}} b))$, and $|a \otimes_{\mathit{rnd}} b| < |a \times b|$. The assumption $|a \otimes_{\mathit{rnd}} b| < |a \times b|$ implies that $\mathit{sign}(a \times b - (a \otimes_{\mathit{rnd}} b)) = \mathit{sign}(a \times b)$. With regards to the sign of $a \otimes_{\mathit{rnd}} b$, one can infer from the condition $a \times b - (a \otimes_{\mathit{rnd}} b) \neq 0$ that $a \times b \neq 0$ and that $a \otimes_{\mathit{rnd}} b = \mathit{rnd}(a \times b)$ given Equation ~7 (Figure ~\ref{fig:background:ops}). Due to the sign preserving properties of faithful rounding functions detailed in Lemma ~\ref{lemma:sign-preservation-proof}, the equality $\mathit{sign}(a \otimes_{\mathit{rnd}} b) = \mathit{sign}(a \times b)$ holds whenever $a \times b - (a \otimes_{\mathit{rnd}} b) \neq0$. Because $\mathit{sign}(a \times b - (a \otimes_{\mathit{rnd}} b)) = \mathit{sign}(a \times b)$, $\mathit{sign}(a \otimes_{\mathit{rnd}} b) = \mathit{sign}(a \times b - (a \otimes_{\mathit{rnd}} b))$ must be true via transitive equality. The resulting equality in signs between $a \otimes_{\mathit{rnd}} b$ and $a \times b - (a \otimes_{\mathit{rnd}} b)$ directly contradicts the earlier assumption that $\mathit{sign}(a \otimes_{\mathit{rnd}} b) \neq \mathit{sign}(a \times b - (a \otimes_{\mathit{rnd}} b))$.
The proposition that $|a \times b| < |a \otimes_{\mathit{rnd}} b|$ implies $\mathit{sign}(a \otimes_{\mathit{rnd}} b) \neq \mathit{sign}(a \times b - (a \otimes_{\mathit{rnd}} b))$ when $a \times b - (a \otimes_{\mathit{rnd}} b) \neq 0$ can also be proven via contradiction. Suppose  $a \times b - (a \otimes_{\mathit{rnd}} b) \neq 0$ and $|a \times b| < |a \otimes_{\mathit{rnd}} b|$ while $\mathit{sign}(a \otimes_{\mathit{rnd}} b) = \mathit{sign}(a \times b - (a \otimes_{\mathit{rnd}} b))$. The inequality $|a \times b | < |a \otimes_{\mathit{rnd}} b|$ would indicate that $\mathit{sign}(a \times b - (a \otimes_{\mathit{rnd}} b)) = \mathit{sign}(-(a \otimes_{\mathit{rnd}} b))$, which implies $a \times b - (a \otimes_{\mathit{rnd}} b)$ and $a \otimes_{\mathit{rnd}} b$ have opposite signs and thus directly contradicts the earlier assumptions. In conclusion, $a \otimes_{\mathit{rnd}} b$ and $a \times b - (a \otimes_{\mathit{rnd}} b)$ cannot have the same signs when $a \times b - (a \otimes_{\mathit{rnd}} b) \neq 0$ and $|a \times b|<|a \otimes_{\mathit{rnd}} b|$.
\end{proof}

Theorem ~\ref{theorem:rz-mult-1-proof} affirms that examining the sign of $a \times b- (a \otimes_{\mathit{rnd}} b)$ relative to $a \otimes_{\mathit{rnd}} b$ is sufficient for assessing whether $|a \times b| < |a \otimes_{\mathit{rnd}} b|$ whenever $a \times b - (a \otimes_{\mathit{rnd}} b) \neq 0$. While the FP multiplication error $a \times b - (a \otimes_{\mathit{rnd}} b)$ can be leveraged to determine whether a FP product $a \otimes_{\mathit{rnd}} b$ needs to be adjusted to match $a \otimes_{\RNZ} b$ (see Equation ~\ref{get-rz-mult}), the error is not guaranteed to be exactly representable as a FP number. As such, Algorithm ~\ref{alg:rz_mult} computes the proxy error term $c1 = \mathit{\mathit{fma}}_{\mathit{rnd}}(a, b, -(a \otimes_{\mathit{rnd}} b))$ for the purposes of applying Theorem ~\ref{theorem:rz-mult-1-proof}. Proving the correctness of $\RZM$ thus entails validating Algorithm ~\ref{alg:rz_mult}'s usage of $c1$ in confirming whether $a \times b - (a \otimes_{\mathit{rnd}} b ) \neq 0$ and $\mathit{sign}(a \times b - (a \otimes_{\mathit{rnd}} b)) \neq \mathit{sign}(a \otimes_{\mathit{rnd}} b)$. 

Algorithm ~\ref{alg:rz_mult} examines whether the FP multiplication error $a \times b - (a \otimes_{\mathit{rnd}} b)$ is not equal to $0$ through the comparison $bit(c1 = \mathit{fma}_{\mathit{rnd}}(a, b, -m)) \neq bit(c2 = \mathit{fma}_{\mathit{rnd}}(-a, b, m))$ where $m = a \otimes_{\mathit{rnd}} b$. We preface our discussion of this comparison by emphasizing that $c1$ and $c2$ are both outputs of fused-multiply-add operations with faithful rounding properties. By definition, $c1$ and $c2$ are each faithfully rounded counterparts to $a \times b - (a \otimes_{\mathit{rnd}} b)$ and $(-(a \times b) + (a \otimes_{\mathit{rnd}} b))$ respectively. Hence, $c1$ is a faithful rounding of the FP multiplication error $a \times b - (a \otimes_{\mathit{rnd}} b)$ and $c2$ is a faithful rounding of the error's negation. In the case of $\RZA$, examining whether the faithful rounding of the FP addition error $a+b - (a \otimes_{\mathit{rnd}} b)$ is neither $+0$ nor $-0$ is guaranteed to be sufficient for verifying whether the real error is a non-zero value. This guarantee is founded on Lemma ~\ref{lemma:rz-add-error-bound}, which establishes that the magnitude of a non-zero FP addition error cannot be less than the smallest positive FP number in the target representation.  Lemma ~\ref{lemma:rz-add-error-bound} does not carry over to FP multiplication, which signifies that the FP multiplication error is susceptible to underflow. As a result, $c1$ could be either $+0$ or $-0$ when the real FP multiplication error is non-zero, and thus checking whether $c1$ is neither $+0$ nor $-0$ is insufficient for assessing whether $a \times b - (a \otimes_{\mathit{rnd}} b) \neq 0$. Given this restriction, we design  $\RZM$ to check whether $bit(c1) \neq bit(c2)$ instead, using Theorem ~\ref{theorem:rz-mult-2-proof} as the justification. We present Theorem ~\ref{theorem:rz-mult-2-proof} again for reference along with its proof.

\begin{theorem}\label{theorem:rz-mult-2-proof}
Let $a$ and $b$ be two non-$NaN$, non-infinity floating-point numbers such that $a \otimes_{\mathit{rnd}} b$ does not overflow for any rounding mode. Let $bit(f)$ be a function that returns the bit-string of any floating- point number $f$. Then, for any $\mathit{rnd}$, $bit(\mathit{fma}_{\mathit{rnd}}(a, b, -(a \otimes_{\mathit{rnd}} b)) \neq bit(\mathit{fma}_{\mathit{rnd}}(-a, b, a \otimes_{\mathit{rnd}} b))$ if and only if $a \times b - (a \otimes_{\mathit{rnd}} b) \neq 0$.
\end{theorem}

\begin{proof}
We first prove that $bit(\mathit{fma}_{\mathit{rnd}}(a, b, -(a
\otimes_{\mathit{rnd}} b)) \neq bit(\mathit{fma}_{\mathit{rnd}}(-a, b,
a \otimes_{\mathit{rnd}} b))$ implies $a \times b-(a
\otimes_{\mathit{rnd}} b)\neq0$ via contradiction. Henceforth, we
refer to $\mathit{fma}_{\mathit{rnd}}(a, b, -(a \otimes_{\mathit{rnd}}
b))$ and $\mathit{fma}_{\mathit{rnd}}(-a, b, a \otimes_{\mathit{rnd}}
b)$ as $c1$ and $c2$ respectively. Assume that $bit(c1) \neq bit(c2)$
and $a \times b-(a \otimes_{\mathit{rnd}} b)=0$ are simultaneously
true for some non-$NaN$ FP numbers $a$ and $b$ such that $a
\otimes_{\mathit{rnd}} b$ does not overflow. We reiterate that given
non-$NaN$, non-infinity operands $a$, $b$, and $c$,
$\mathit{fma}_{\mathit{rnd}}(a, b, c)$ returns a faithful rounding of
the real value $a \times b + c$. Specifically,
$\mathit{fma}_{\mathit{rnd}}(a , b, c)$ performs a faithfully rounded
FP addition between $a \times b$ and $c$ with no intermediate rounding
for the multiplication. When $a \times b - (a \otimes_{\mathit{rnd}}
b) = 0$, $a \times b = a \otimes_{\mathit{rnd}} b$ and thus $c1 = (a
\otimes_{\mathit{rnd}} b) \oplus_{\mathit{rnd}} (-(a
\otimes_{\mathit{rnd}} b))$. Under the assumption that $a \times b -
(a \otimes_{\mathit{rnd}} b) = 0$, $(-a) \times b = -(a \times b) =
-(a \otimes_{\mathit{rnd}} b)$. As such, $c2$ would be equal to $(-(a
\otimes_{\mathit{rnd}} b)) \oplus_{\mathit{rnd}} (a
\otimes_{\mathit{rnd}} b)$ when $a \times b - (a
\otimes_{\mathit{rnd}} b) = 0$. In essence, $a \times b - (a
\otimes_{\mathit{rnd}} b) = 0$ implies that $c1$ and $c2$ are the
respective outputs of two FP additions with symmetric operands across
all rounding modes. Given the commutativity of addition, this would
indicate that $c1 = c2$. Because all non-$NaN$ FP numbers have unique
bit-string representations, $a \times b - (a \otimes_{\mathit{rnd}} b)
= 0$ and $bit(c1) \neq bit(c2)$ cannot simultaneously be true.

The proposition that $a \times b- (a \otimes_{\mathit{rnd}} b)\neq0$
implies $bit(c1) \neq bit(c2)$ can also be proven via
contradiction. Assume that there exist two non-$NaN$, non-infinity FP
numbers $a$ and $b$ such that $a \otimes_{\mathit{rnd}} b$ does not
induce overflow for any rounding mode, $a \times b - (a
\otimes_{\mathit{rnd}} b) \neq 0$, and $bit(c1) = bit(c2)$. As
previously explained, $c1$ and $c2$ are faithful roundings of $a
\times b - (a \otimes_{\mathit{rnd}} b)$ and $-(a \times b ) + (a
\otimes_{\mathit{rnd}} b)$ respectively. Given the sign preserving
nature of faithful rounding with respect to non-zero numbers detailed
in Lemma ~\ref{lemma:sign-preservation-proof}, $a \times b - (a
\otimes_{\mathit{rnd}} b) \neq 0$ implies $\mathit{sign}(c1) =
\mathit{sign}(a \times b - (a \otimes_{\mathit{rnd}} b))$. The
assumption that $a \times b - (a \otimes_{\mathit{rnd}} b) \neq 0$
also implies that $-(a \times b) + (a \otimes_{\mathit{rnd}} b) \neq
0$, and thus Lemma ~\ref{lemma:sign-preservation-proof} leads to the
equality $\mathit{sign}(c2) = \mathit{sign}(-(a \times b) + (a
\otimes_{\mathit{rnd}} b))$. Because $a \times b - (a
\otimes_{\mathit{rnd}} b)$ and its negation $-(a \times b) + (a
\otimes_{\mathit{rnd}} b)$ are two non-zero real numbers with opposite
signs when $a \times b - (a \otimes_{\mathit{rnd}} b) \neq 0$,
$\mathit{sign}(c1)$ and $\mathit{sign}(c2)$ cannot be the same under
such a condition. The signs of $c1$ and $c2$ derivable under the
condition $a \times b - (a \otimes_{\mathit{rnd}} b) \neq 0$
contradicts the assumption that $bit(c1) = bit(c2)$ because the
bit-strings of $c1$ and $c2$ will differ in terms of the leading
bit. Therefore, $a \times b - (a \otimes_{\mathit{rnd}} b) \neq 0$ and
$bit(c1) = bit(c2)$ cannot simultaneously be true.
\end{proof}

Through Theorem ~\ref{theorem:rz-mult-2-proof}, we guarantee that the
condition $bit(c1) \neq bit(c2)$ in line 5 of Algorithm
~\ref{alg:rz_mult} is sufficient for inspecting whether $a \times b -
(a \otimes_{\mathit{rnd}} b) \neq 0$. The last step in proving the
correctness of $\RZM$ is confirming that the comparison
$(bit(c1)\ \textbf{xor} \ bit(m)) \geq 0x8000000000000000$, which
checks whether $\mathit{sign}(c1) \neq \mathit{sign}(a
\otimes_{\mathit{rnd}} b)$, is sufficient for assessing whether
$\mathit{sign}(a \times b - (a \otimes_{\mathit{rnd}} b)) \neq
\mathit{sign}(a \otimes_{\mathit{rnd}} b)$. We justify $\RZM$'s usage
of the condition $(bit(c1)\ \textbf{xor} \ bit(m)) \geq
0x8000000000000000$ by establishing the following property for $c1$:
$\mathit{sign}(c1) = \mathit{sign}(a \times b - (a
\otimes_{\mathit{rnd}} b))$ whenever $a \times b - (a
\otimes_{\mathit{rnd}} b) \neq 0$. Given the definition of
$\mathit{fma}_{\mathit{rnd}}$, $c1$ is a faithful rounding of $a
\times b - (a \otimes_{\mathit{rnd}} b)$. In the case that $a \times b
- (a \otimes_{\mathit{rnd}} b) \neq 0$, $c1$ would be a faithfully
rounding of a non-zero real number, and thus $c1 = \mathit{rnd}(a
\times b - (a \otimes_{\mathit{rnd}} b))$ for any rounding function
$\mathit{rnd}$ with faithful rounding properties. Consequently, one
can apply Lemma ~\ref{lemma:sign-preservation-proof} to conclude that
$\mathit{sign}(c1) = \mathit{sign}(a \times b - (a
\otimes_{\mathit{rnd}} b))$ when $a \times b - (a
\otimes_{\mathit{rnd}} b) \neq 0$.

Having established the key implications of the proxy FP multiplication error $c1$, we now summarize how the properties of $c1$ and $c2$ guarantee that $\RZM$ can accurately apply Theorem ~\ref{theorem:rz-mult-1-proof} through the conditions $bit(c1) \neq bit(c2)$ and $(bit(c1)\ \textbf{xor} \ bit(m)) \geq 0x8000000000000000$ in line 5 of Algorithm ~\ref{alg:rz_mult}. Theorem ~\ref{theorem:rz-mult-2-proof} establishes that the condition $bit(c1) \neq bit(c2)$ is sufficient for testing whether $a \times b - (a \otimes_{\mathit{rnd}} b) \neq 0$. $\RZM$ can thus guarantee that $a \times b - (a \otimes_{\mathit{rnd}} b) \neq 0$  when the condition $bit(c1) \neq bit(c2)$ is true. Because  $a \times b - (a \otimes_{\mathit{rnd}} b) \neq 0$ implies $\mathit{sign}(c1) = \mathit{sign}(a \times b - (a \otimes_{\mathit{rnd}} b))$, the comparison $(bit(c1)\ \textbf{xor} \ bit(m)) \geq 0x8000000000000000$, which checks if $\mathit{sign}(c1) \neq \mathit{sign}(a \otimes_{\mathit{rnd}} b)$, is an appropriate test for $\mathit{sign}(a \times b - (a \otimes_{\mathit{rnd}} b)) \neq \mathit{sign}(a \otimes_{\mathit{rnd}} b)$ when $bit(c1) \neq bit(c2)$ is true. Therefore, in the case that both conditions $bit(c1) \neq bit(c2)$ and $(bit(c1)\ \textbf{xor} \ bit(m)) \geq 0x8000000000000000$ are true, $\RZM$ can correctly assume that $a \times b - (a \otimes_{\mathit{rnd}} b) \neq 0$ and $\mathit{sign}(a \times b - (a \otimes_{\mathit{rnd}} b)) \neq \mathit{sign}(a \otimes_{\mathit{rnd}} b)$. $\RZM$ can subsequently conclude that $|a \times b| < |a \otimes_{\mathit{rnd}} b|$ based on Theorem ~\ref{theorem:rz-mult-1-proof} and apply the necessary adjustment to $a \otimes_{\mathit{rnd}} b$. In conclusion, Theorems ~\ref{theorem:rz-mult-1-proof} and ~$\ref{theorem:rz-mult-2-proof}$ in conjunction with Lemma ~\ref{lemma:sign-preservation-proof} confirm that $\RZM$ can augment $a \otimes_{\mathit{rnd}} b$ in accordance with Equation~\ref{get-rz-mult} to produce $a \otimes_{RZ} b$ under all rounding modes.

\end{document}